%% file: main_arxiv.tex
\documentclass[journal]{IEEEtran}
\usepackage{amsthm,amssymb,amsmath,rangecite}
\usepackage{commath,amsmath,graphicx,epstopdf,amsthm,amssymb,float,cite,color,array,caption,subcaption}
\usepackage[ruled,vlined]{algorithm2e}
\newtheorem{Theorem}{Theorem}

\newtheorem{proposition}[Theorem]{Proposition}
\newtheorem{definition}{Definition}
\newtheorem{example}{Example}
\newtheorem{claim}[Theorem]{Claim}

\usepackage{lscape}
\usepackage{cite}
\usepackage{url}
\usepackage{hyperref}

\input{myshorts}

\usepackage{ifthen}
\usepackage{comment}
\begin{document}
\title{Barankin-Type Bound for Constrained Parameter Estimation}
\author{Eyal~Nitzan,~\IEEEmembership{Member,~IEEE,}
Tirza~Routtenberg,~\IEEEmembership{Senior~Member,~IEEE,}
and~Joseph~Tabrikian,~\IEEEmembership{Fellow,~IEEE}
\thanks{This research was partially supported by THE ISRAEL SCIENCE FOUNDATION (grant No. 2666/19 and grant No. 1148/22).}

\thanks{{\footnotesize{
E. Nitzan, T. Routtenberg, and J. Tabrikian are with the School of Electrical and Computer Engineering, Ben-Gurion University of the Negev, Beer-Sheva, Israel, e-mails: tirzar@bgu.ac.il, joseph@bgu.ac.il. \\
© 2023 IEEE. Personal use of this material is permitted. Permission from
IEEE must be obtained for all other uses, in any current or future media,
including reprinting/republishing this material for advertising or promotional
purposes, creating new collective works, for resale or redistribution to servers
or lists, or reuse of any copyrighted component of this work in other works.
}}}
}

\maketitle
\nopagebreak

\begin{abstract}
In constrained parameter estimation, the classical constrained Cram$\acute{\text{e}}$r–Rao bound (CCRB) and the recent Lehmann-unbiased CCRB (LU-CCRB) are lower bounds on the performance of mean-unbiased and Lehmann-unbiased estimators, respectively. Both the CCRB and the LU-CCRB require differentiability of the likelihood function, which can be a restrictive assumption. Additionally, these bounds are local bounds that are inappropriate for predicting the threshold phenomena of the constrained maximum likelihood (CML) estimator. The constrained Barankin-type bound (CBTB) is a nonlocal mean-squared-error (MSE) lower bound for constrained parameter estimation that does not require differentiability of the likelihood function. However, this bound requires a restrictive mean-unbiasedness condition in the constrained set. In this work, we propose the Lehmann-unbiased CBTB (LU-CBTB) on the weighted MSE (WMSE). This bound does not require differentiability of the likelihood function and assumes uniform Lehmann-unbiasedness, which is less restrictive than the CBTB uniform mean-unbiasedness. We show that the LU-CBTB is tighter than or equal to the LU-CCRB and coincides with the CBTB for linear constraints. For nonlinear constraints the LU-CBTB and the CBTB are different and the LU-CBTB can be a lower bound on the WMSE of constrained estimators in cases, where the CBTB is not. In the simulations, we consider direction-of-arrival estimation of an unknown constant modulus discrete signal. In this case, the likelihood function is not differentiable and constrained Cram$\acute{\text{e}}$r–Rao-type bounds do not exist, while CBTBs exist. It is shown that the LU-CBTB better predicts the CML estimator performance than the CBTB, since the CML estimator is Lehmann-unbiased but not mean-unbiased.
\end{abstract}

\section{Introduction} \label{sec:Introduction}
Various signal processing applications require estimation of parameter vectors under given parametric constraints \cite{Hero_constraint,Stoica_Ng,HERO_EMISSION,SADLER_sim,WIJNHOLDS_sim,ROUTTENBERG_sim,HERO_HYPER,MENNI_sim,SHIR_PARAMETRIC,PREVOST_CONSTRAINED_CRAMER}.
Lower bounds for non-Bayesian constrained parameter estimation are useful tools for performance analysis of estimators and for system design. The constrained Cram$\acute{\text{e}}$r-Rao bound (CCRB) \cite{Hero_constraint,Marzetta,Stoica_Ng,CCRB_complex,BenHaim,sparse_con,Moore_fitting,Moore_scoring,normC} is the most commonly-used performance bound in the constrained setting. However, it was recently shown in \cite{SOMEKH_LESHEM,JOURNAL_CONSTRAINTS_EYAL,SSP_EYAL} that the unbiasedness conditions of the CCRB are too restrictive and thus, it may be non-informative outside the asymptotic region. Less restrictive unbiasedness conditions for constrained parameter estimation, named C-unbiasedness conditions, were derived in \cite{JOURNAL_CONSTRAINTS_EYAL} by using the Lehmann-unbiasedness concept \cite{point_est,SAM2012constraint,PCRB,CYCLIC,BAR1,SELECTION,BAR_NB,SSP_EYAL,LETTER_NORM_EYAL,MEIR_ROUTTENBERG}. In addition, the Lehmann-unbiased CCRB (LU-CCRB) on the weighted mean-squared-error (WMSE) of locally C-unbiased estimators, was proposed in \cite{JOURNAL_CONSTRAINTS_EYAL}. It was shown, in some simulations, that in nonlinear constrained parameter estimation, the LU-CCRB provides a lower bound for the constrained maximum likelihood (CML) estimator performance, while the CCRB does not.\\
\indent
Cram$\acute{\text{e}}$r-Rao-type bounds suffer from two main drawbacks \cite{HAMMERSLEY,HCR,TTB1}. First, they require differentiability of the likelihood function that can be restrictive in some cases \cite{HASSIBI_BOYD,LAROSA_CHANGE_POINT}. Second, these bounds are based on local statistical information of the observations in the vicinity of the true parameters. Therefore, they may not be tight for low signal-to-noise ratio (SNR) or small number of observations. Barankin-type bounds (BTBs) \cite{BARANKIN,HAMMERSLEY,HCR,MS,REUVEN_MESSER,TTB1} can be used to provide informative performance benchmarks in cases, where Cram$\acute{\text{e}}$r-Rao-type bounds do not exist or are not sufficiently tight. A notable BTB is the Hammersley-Chapman-Robbins (HCR) bound \cite{HAMMERSLEY,HCR}, which is based on enforcing uniform mean-unbiasedness over the parameter space and maximizing the bound w.r.t. test-point vectors. Based on this BTB, the constrained BTB (CBTB) was proposed in \cite{Hero_constraint}, where the test-point vectors of the CBTB are taken only from the constrained set. However, this bound was not implemented directly in \cite{Hero_constraint} and was only used to obtain the CCRB as a special case. CBTB for the special case of sparse parameter estimation in linear models was derived in \cite{JUNG_SPARSE} under the assumption of mean-unbiasedness in the set of sparse vectors. Thus, the existing CBTBs require mean-unbiasedness in the constrained set that can be too restrictive for practical estimators \cite{SOMEKH_LESHEM,JOURNAL_CONSTRAINTS_EYAL}. Consequently, a CBTB under mild unbiasedness conditions can be very useful.\\
\indent
In this paper, we derive a novel CBTB on the WMSE. This bound, named Lehmann-unbiased CBTB (LU-CBTB), is derived under the C-unbiasedness conditions from \cite{JOURNAL_CONSTRAINTS_EYAL}. The LU-CBTB is derived by using test-point vectors taken only from the constrained set and by enforcing the C-unbiasedness condition at these test-point vectors. The C-unbiasedness conditions required by the LU-CBTB are less restrictive than the mean-unbiasedness conditions required by the CBTB from \cite{Hero_constraint}. Thus, there may be cases where the LU-CBTB is a valid lower bound on the WMSE of constrained estimators, such as the CML estimator, while the CBTB is higher than the WMSE of these estimators. Given the test-point vectors, we show that the CBTB on the WMSE is higher than or equal to the proposed LU-CBTB and that the LU-CBTB coincides with the CBTB under linear constraints. It is also shown that the LU-CBTB is always tighter than or equal to the LU-CCRB from \cite{JOURNAL_CONSTRAINTS_EYAL}, since the LU-CCRB can be obtained as a special case of the LU-CBTB. In the simulations, we consider the problem of direction-of-arrival (DOA) estimation of an unknown constant modulus (CM) discrete signal. In this case, the likelihood function is not differentiable with respect to (w.r.t.) some of the unknown parameters and constrained Cram$\acute{\text{e}}$r–Rao-type bounds cannot be used. The LU-CBTB and the CBTB are evaluated and compared to the WMSE of the CML estimator, which is shown to be a C-unbiased estimator but not mean-unbiased. Thus, in the considered example the LU-CBTB is a lower bound on the CML WMSE, while the CBTB is not. In addition, it is shown that the LU-CBTB better captures the CML error behavior than the CBTB.\\
\indent
The WMSE is a scalar risk measure for multiparameter estimation. In contrast to the CCRB and CBTB, which are matrix lower bounds, the proposed bound cannot be formulated as an MSE matrix inequality. A matrix lower bound is useful, because it provides a lower bound for estimating any linear combination of the parameter vector. Specifically, it allows to obtain a bound for any of the elements of the parameter vector. In the literature, there exist several scalar bounds, (see e.g. \cite{ziv1969some,PCRB,CYCLIC,Aharon_Tabrikian}). In such cases, one may derive a bound for any given linear combination of the MSE matrix, and the bound may depend on the linear combination matrix. Although it results in a scalar bound, it allows to obtain a bound on the MSE of {\it any} linear combination of the parameter vector or any linear combination of the MSE matrix, such as the trace of the MSE matrix. This approach was adopted in \cite{EXTENDED_ZZ,JOURNAL_CONSTRAINTS_EYAL}. 

\indent
Throughout this paper, we denote vectors by boldface lowercase letters and matrices by boldface uppercase letters. The $m$th element of the vector $\avec$, a subvector of $\avec$ with indices $l,l+1,\ldots,l+M$, and the $(m,n)$th element of the matrix $\Amat$ are denoted by $a_m$, $[\avec]_{l:l+M}$, and $A_{m,n}$, respectively. The identity matrix of dimension $M\times M$ is denoted by $\Imat_M$, its $m$th column is denoted by $\evec_m^{(M)}$, $\forall m=1,\ldots,M$, and $\onevec$ denotes a vector of all ones.
The notation $\zerovec$ stands for vector/matrix of zeros. The notations ${\rm{Tr}}(\cdot)$ and ${\rm{vec}}(\cdot)$ denote the trace and vectorization operators, where the vectorization operator stacks the columns of its input matrix into a column vector. The notation ${\rm{diag}}(\cdot)$ is a diagonal matrix, whose diagonal elements are given by the arguments. The notations $(\cdot)^T$, $(\cdot)^{-1}$, and $(\cdot)^{\dagger}$ denote the transpose, inverse, and Moore-Penrose pseudo-inverse, respectively. The square root of a positive semidefinite matrix, $\Amat$, is denoted by $\Amat^{\frac{1}{2}}$. The column and null spaces of a matrix are denoted by ${\mathcal{R}}(\cdot)$ and ${\mathcal{N}}(\cdot)$, respectively. The matrices $\Pmat_\Amat=\Amat\Amat^{\dagger}=\Amat(\Amat^T\Amat)^{\dagger}\Amat^T$ and $\Pmat_\Amat^\bot=\Imat_M-\Pmat_\Amat$ are the orthogonal projection matrices onto ${\mathcal{R}}(\Amat)$ and ${\mathcal{N}}(\Amat^T)$, respectively \cite{CAMPBELL}. The notation $\Amat\otimes\Bmat$ is the Kronecker product of the matrices $\Amat$ and $\Bmat$. The gradient of a vector function $\gvec$ of $\thetavec$, $\nabla_{\thetavecsmall}\gvec(\thetavec)$, is a matrix whose $(m,n)$th element is $\frac{{\partial}g_m(\thetavecsmall)}{{\partial}\theta_n}$. The real and imaginary parts of an argument are denoted by $\text{Re}\{\cdot\}$ and $\text{Im}\{\cdot\}$, respectively, and $j\define{\sqrt{-1}}$. The notation $\angle{\cdot}$ stands for the phase of a complex scalar, which is assumed to be restricted to the interval $[-\pi,\pi)$.\\
\indent
The remainder of the paper is organized as follows. In Section \ref{sec:constrained estimation Problem}, we present the constrained estimation model and relevant background for this paper. The LU-CBTB and its properties are derived in Section \ref{sec:LU-CBTB}. Our simulations appear in Section \ref{sec:Simulations} and in Section \ref{sec:Conclusion} we give our conclusions.

\section{Constrained estimation problem} \label{sec:constrained estimation Problem}
This section provides the necessary background and definitions required for our main contribution, which is a new lower bound presented in the next section. We first introduce the model of constrained estimation in Subsection \ref{model_subsec}. In Subsection \ref{Classical_bounds_subsec}, we discuss the CBTB and CCRB on the MSE matrix and their adaptation to WMSE lower bounds. In Subsection \ref{unbiasedness_subsec}, we address C-unbiasedness (also known as Lehmann-unbiasedness under constraints) \cite{JOURNAL_CONSTRAINTS_EYAL}. Finally, in Subsection \ref{LU_CCRB_subsec}, we present the LU-CCRB on the WMSE, which was developed in our previous work \cite{JOURNAL_CONSTRAINTS_EYAL}.
\subsection{Constrained model}
\label{model_subsec}
Let $(\Omega_\xvec,{\cal{F}},P_\thetavecsmall)$ denote a probability space, where $\Omega_\xvec$ is the observation space, ${\cal{F}}$ is the $\sigma$-algebra on $\Omega_\xvec$, and $\left\{P_\thetavecsmall \right\}$ is a family of probability measures parameterized by the deterministic unknown parameter vector $\thetavec\in{\mathbb{R}}^M$. Each probability measure, $P_\thetavecsmall$, is assumed to have an associated probability density function, $f_\xvec(\xvec;\thetavec)$, whose support w.r.t. $\xvec$ is independent of $\thetavec$. The expectation operator w.r.t. $P_\thetavecsmall$ is denoted by ${\rm{E}}[\cdot;\thetavec]$.\\
\indent
We suppose that $\thetavec$ is restricted to the set
\be \label{set_def}
\Theta_\fvec=\{\thetavec\in{\mathbb{R}}^M:\fvec(\thetavec)=\zerovec\},
\ee
where $\fvec:{\mathbb{R}}^M\rightarrow {\mathbb{R}}^K $ is a continuously differentiable vector-valued function. It is assumed that $0<K<M$ and that the matrix $\Fmat(\thetavec)\define\nabla_{\thetavecsmall}\fvec(\thetavec)\in{\mathbb{R}}^{K\times M}$ has a rank $K$ for any $\thetavec\in\Theta_\fvec$, i.e. the constraints are not redundant. Thus, for any $\thetavec\in\Theta_\fvec$ there exists a matrix $\Umat(\thetavec)\in{\mathbb{R}}^{M\times(M-K)}$, such that
\be \label{one}
\Fmat(\thetavec)\Umat(\thetavec)=\zerovec
\ee
and
\be \label{two}
\Umat^T(\thetavec)\Umat(\thetavec)=\Imat_{M-K}.
\ee
An estimator of $\thetavec$ based on a random observation vector $\xvec\in\Omega_\xvec$ is denoted by $\hat{\thetavec}$, where $\hat{\thetavec}:\Omega_\xvec\rightarrow{\mathbb{R}}^M$.\\
\indent
MSE lower bounds on vector parameters are usually presented in a matrix inequality form. The advantage of such a matrix bound is that it provides a lower bound for {\it any} linear combination of the estimation error vector. Such a bound can be useful if, under some unbiasedness restrictions, the MSE matrix of the optimal estimator can lower bound the MSE matrix of any estimator. However, in many cases, such an optimal estimator in a matrix sense does not exist. This problem can be handled by optimizing the MSE of a given linear combination of the parameter vector. In such cases, a matrix lower bound, which is independent of the linear combination, cannot accurately predict the MSE matrix of any estimator. In order to avoid compromising the tightness of the bound, in this paper, we consider a weighted squared-error (WSE) cost function \cite{JOURNAL_CONSTRAINTS_EYAL,PILZ,ELDAR_WEIGHTED_MSE},
\begin{equation}\label{WSE}
C_{\text{WSE}}(\hat{\thetavec},\thetavec)\define(\hat{\thetavec}-\thetavec)^T\Wmat(\hat{\thetavec}-\thetavec),
\end{equation}
where $\Wmat\in\mathbb{R}^{M\times M}$ is a symmetric, positive semidefinite weighting matrix.
The WMSE risk is obtained by taking the expectation of \eqref{WSE} and is given by
\be\label{WMSE}
{\text{WMSE}}_{\hat{\thetavecsmall}}(\thetavec)\define{\rm{E}}[C_{\text{WSE}}(\hat{\thetavec},\thetavec);\thetavec]={\rm{E}}[(\hat{\thetavec}-\thetavec)^T\Wmat(\hat{\thetavec}-\thetavec);\thetavec].
\ee
The WMSE \cite{PILZ,ELDAR_WEIGHTED_MSE} is a family of scalar risks for estimation of an unknown parameter vector, where for each $\Wmat$ we obtain a different risk.  Therefore, the WMSE 
allows flexibility in the design of estimators and the derivation of performance bounds. 
 For example,  for $\Wmat=\Imat_M$ 
we obtain the special case of the trace of the MSE matrix criterion. Another example is when one may wish to consider the estimation of each element of the unknown parameter vector separately. Moreover, $\Wmat$ can compensate for possibly different units of the parameter vector elements. Another example is the estimation in the presence of nuisance parameters, where we are only interested in the MSE for estimation of a subvector of the unknown parameter vector (see e.g. \cite{SSP_EYAL}, \cite[p. 461]{point_est}) and thus, $\Wmat$ includes zero elements for the nuisance parameters, and ones for the parameters of interest. For example, if one is only interested in estimation $\theta_1$, then $\Wmat={\rm{diag}}([1,0,\ldots,0]^T)$.
Finally, by taking $\Wmat=\avec \avec^T$, where $\avec\in\mathbb{R}^{M}$ is an arbitrary vector, we can obtain any linear combination of the estimation errors, since in this case, the WMSE from \eqref{WMSE} is related to  the MSE matrix:
\begin{align}
\label{WMSE_for_aa}
{\text{WMSE}}_{\hat{\thetavecsmall}}(\thetavec)=& {\rm{E}}[(\hat{\thetavec}-\thetavec)^T\avec\avec^T(\hat{\thetavec}-\thetavec);\thetavec]\nonumber \\
= &\avec^T {\rm{E}}\left[(\hat{\thetavec}-\thetavec) (\hat{\thetavec}-\thetavec)^T;\thetavec \right] \avec.
\end{align}

\subsection{Background: conventional CBTB and CCRB}
\label{Classical_bounds_subsec}
In this subsection, we present the conventional CBTB and CCRB on the MSE matrix and their adaptation to the WMSE. 

Let $\thetavec_1,\ldots,\thetavec_P\in\mathbb{R}^M$ denote $P$ test-point vectors 
and define the test-point matrix 
\be\label{test_point_matrix}
\Pimat\define[\thetavec_1,\ldots,\thetavec_P]\in\mathbb{R}^{M\times P}. 
\ee
The HCR-based CBTB from \cite{Hero_constraint} is a lower bound on the MSE of mean-unbiased estimators where the mean-unbiasedness is assumed uniformly over the constrained set. This lower bound is defined as the conventional BTB with test-point vectors that are taken only from the constrained parameter space, $\Theta_\fvec$, rather than the entire parameter space, $\mathbb{R}^M$. Thus, we consider $\Pimat\in\Theta_\fvec^P$ rather than $\mathbb{R}^{M\times P}$ and the CBTB is given by

\beqna\label{CBTBmat}
\Bmat_{\text{CBTB}}(\thetavec)\hspace{6cm}\nonumber\\
\define\underset{\Pimatsmall\in\Theta_\fvec^P}{\sup}\left\{\Tmat(\thetavec,\Pimat)\left(\Bmat(\thetavec,\Pimat)-\onevec \onevec^T\right)^{\dagger}\Tmat^T(\thetavec,\Pimat)\right\},
\eeqna
where 
\be\label{T_BTB}
\Tmat(\thetavec,\Pimat)\define[\thetavec_1-\thetavec,\ldots,\thetavec_P-\thetavec]\in\mathbb{R}^{M\times P}.
\ee
The $( m,n)$th element of the matrix $\Bmat(\thetavec,\Pimat)\in\mathbb{R}^{P\times P}$ is given by
\be\label{Bmat_BTB_elements}
B_{ m,n}(\thetavec,\Pimat)={\rm{E}}\left[\frac{f_\xvec(\xvec;\thetavec_m)}{f_\xvec(\xvec;\thetavec)}\frac{f_\xvec(\xvec;\thetavec_n)}{f_\xvec(\xvec;\thetavec)};\thetavec\right],
\ee
$\forall m,n=1,\ldots,P$.

Let us define the gradient of the log-likelihood function 
\be\label{l_define}
\upsilonvec(\xvec,\thetavec)\define\nabla_{\thetavecsmall}^T\log f_\xvec(\xvec;\thetavec)
\ee
and the Fisher information matrix
\be\label{FIM}
\Jmat(\thetavec)\define{\rm{E}}\left[\upsilonvec(\xvec,\thetavec)\upsilonvec^T(\xvec,\thetavec);\thetavec\right]. 
\ee
Then, the CCRB is given by \cite{Stoica_Ng,BenHaim}
\be \label{CCRBmat}
\Bmat_{\text{CCRB}}(\thetavec)\define\Umat(\thetavec)\left(\Umat^T(\thetavec)\Jmat(\thetavec)\Umat(\thetavec)\right)^{\dagger}\Umat^T(\thetavec).
\ee 
The existence of the CCRB requires differentiability of the likelihood function that may not be satisfied \cite{HAMMERSLEY,HASSIBI_BOYD}. If the CCRB exists, then it can be obtained as a special case of the CBTB \cite{Hero_constraint}. The CBTB and the CCRB from \eqref{CBTBmat} and \eqref{CCRBmat}, respectively, are lower bounds for estimators that have zero mean-bias in the constrained set, which can be very restrictive and may not be satisfied by the commonly-used CML estimator, as shown in \cite{JOURNAL_CONSTRAINTS_EYAL,SOMEKH_LESHEM,LETTER_NORM_EYAL}. 

As any MSE matrix lower bound, the CBTB and the CCRB can be formulated as lower bounds on the WMSE from \eqref{WMSE} for any symmetric, positive semidefinite weighting matrix, $\Wmat$:
\be\label{CMS_W_WMSE}
{\text{WMSE}}_{\hat{\thetavecsmall}}(\thetavec)\geq B_{\text{CBTB}}(\thetavec,\Wmat)
\ee
and
\be\label{CCRB_W_WMSE}
{\text{WMSE}}_{\hat{\thetavecsmall}}(\thetavec)\geq B_{\text{CCRB}}(\thetavec,\Wmat),
\ee
respectively, where
\beqna\label{CMS_W}
B_{\text{CBTB}}(\thetavec,\Wmat)\define\hspace{5.25cm}
\nonumber\\\underset{\Pimatsmall\in\Theta_\fvec^P}{\sup}{\rm{Tr}}\bigg(\Tmat(\thetavec,\Pimat)
\left(\Bmat(\thetavec,\Pimat)-\onevec \onevec^T\right)^{\dagger}
\Tmat^T(\thetavec,\Pimat)\Wmat\bigg)
\eeqna
and
\be \label{CCRB_W}
\begin{split}
&B_{\text{CCRB}}(\thetavec,\Wmat)\define{\text{Tr}}\left(\Bmat_{\text{CCRB}}(\thetavec)\Wmat\right)\\
&={\text{Tr}}\left(\left(\Umat^T(\thetavec)\Jmat(\thetavec)\Umat(\thetavec)\right)^{\dagger}\left(\Umat^T(\thetavec)\Wmat\Umat(\thetavec)\right)\right).
\end{split}
\ee
\indent
\subsection{C-unbiasedness}
\label{unbiasedness_subsec}
In the following, we present the C-unbiasedness definition from our previous work \cite{JOURNAL_CONSTRAINTS_EYAL}, which is a Lehmann-unbiasedness \cite{point_est} under the WMSE risk and the parametric equality constraints from \eqref{set_def}.
In \cite{JOURNAL_CONSTRAINTS_EYAL} we used
the {\it local} C-unbiasedness condition for the development of the  LU-CCRB. In this paper, we use
the {\it uniform} C-unbiasedness condition for the development of the LU-CBTB.
The C-unbiasedness is less restrictive than the mean-unbiasedness and may be satisfied by practical estimators even in nonlinear constrained estimation problems and outside the asymptotic region.

Lehmann \cite{point_est,LEHMANN_CONCEPT} proposed a generalization of the unbiasedness concept, which depends on the considered cost function and on the parameter space, as presented in the following definition.
\begin{definition}
\label{unbiased_definition}
The estimator $\hat{\thetavec}$ is said to be a uniformly unbiased estimator of $\thetavec$ in the Lehmann sense \cite{point_est,LEHMANN_CONCEPT} w.r.t. the cost function $C(\hat{\thetavec},\thetavec)$ 
if 
\be \label{defdef}
{\rm{E}}[C(\hat{\thetavec},\etavec);\thetavec] \geq {\rm{E}}[C(\hat{\thetavec},\thetavec);\thetavec],~\forall \etavec,\thetavec\in \Omega_\thetavecsmall,
\ee
where $\Omega_\thetavecsmall$ is the considered parameter space.
\end{definition}
The Lehmann-unbiasedness definition implies that an estimator is unbiased if, on average, it is  ``closer''  to the true parameter, $\thetavec$, than to any other value in the parameter space, $\etavec\in\Omega_\thetavecsmall$. The measure of closeness between the estimator and the parameter is the cost function, $C(\hat{\thetavec},\thetavec)$. For example, in \cite{LEHMANN_CONCEPT}, it is shown that under the scalar squared-error cost function, $C(\hat{\theta},\theta)=(\hat{\theta}-\theta)^2$, the Lehmann-unbiasedness in \eqref{defdef} is reduced to the conventional mean-unbiasedness,
${\rm{E}}[\hat{\theta}-\theta]=0$, $\forall\theta\in {\Omega_\theta}$. Lehmann-unbiasedness conditions have been used in the literature with various cost functions (see, e.g., \cite{LEHMANN_CONCEPT,PCRB,BAR1,CYCLIC,SELECTION}).
In particular, for the constrained parameter estimation problem described in Subsection \ref{model_subsec} the WSE cost function from \eqref{WSE}, the Lehmann unbiasedness can be described by the C-unbiasedness, described in the following proposition.
\begin{proposition}\label{Cunbias_prop}
Given WSE cost function with a positive semidefinite weighting matrix $\Wmat\in\mathbb{R}^{M\times M}$, a necessary condition for the estimator $\hat{\thetavec}:\Omega_\xvec\rightarrow{\mathbb{R}}^M$ to be a uniformly unbiased estimator of $\thetavec\in{\mathbb{R}}^M$ in the Lehmann sense under the constrained set in \eqref{set_def} is
\be \label{unbiased_cond_nec}
\Umat^T(\thetavec)\Wmat {\rm{E}}[\hat{\thetavec}-\thetavec;\thetavec]=\zerovec,~\forall \thetavec\in\Theta_\fvec,
\ee
\end{proposition} 
\begin{proof} 
The full detailed proof of this proposition can be found in our previous work (see proof of Proposition 1 in \cite{JOURNAL_CONSTRAINTS_EYAL}), where it is shown that by 
substituting $\Omega_\thetavecsmall=\Theta_\fvec$ and the WSE cost function from \eqref{WSE} in \eqref{defdef}, and using tools from constrained minimization  \cite{BETTS},  one obtains that \eqref{unbiased_cond_nec} is the uniformly Lehmann unbiasedness for constrained parameter estimation.
\end{proof}
In the following, 
\color{black}
given a positive semidefinite weighting matrix $\Wmat\in\mathbb{R}^{M\times M}$, 
we say that the estimator $\hat{\thetavec}$ is a uniformly C-unbiased estimator of $\thetavec\in\Theta_\fvec$ if  it satisfies 
\eqref{unbiased_cond_nec}.

\noindent
The columns of the matrix $\Umat(\thetavec)$ span the feasible directions of the constrained set \cite{JOURNAL_CONSTRAINTS_EYAL,BenHaim}. Thus, the C-unbiasedness definition from \eqref{unbiased_cond_nec} implies that at any point, $\thetavec\in\Theta_\fvec$, only the components of the weighted bias vector, $\Wmat{\rm{E}}[\hat{\thetavec}-\thetavec;\thetavec]$, in the feasible directions of the constrained set must be zero, rather than the entire weighted bias vector.
It can be seen that if an estimator has zero mean-bias in the constrained set, {\it i.e.} ${\rm{E}}[\hat{\thetavec}-\thetavec;\thetavec]=\zerovec$, $\forall\thetavec\in \Theta_\fvec$, then it satisfies \eqref{unbiased_cond_nec} but not vice versa. Thus, the uniform C-unbiasedness condition is a weaker condition than requiring uniform mean-unbiasedness in the constrained set.
\\
\indent
\begin{example} 
For the special case of an unconstrained estimation problem, in which  $K=0$,  the matrix $\Umat(\thetavec)$ from \eqref{one} and \eqref{two} is an identity matrix, i.e. $\Umat(\thetavec)=\Imat_M$, and the parameter space is $\Theta_\fvec={\mathbb{R}}^M$.
If we use $\Wmat=\Imat_M$,  where the WMSE is reduced to the trace of the MSE matrix, then the uniform C-unbiasedness requirement in \eqref{unbiased_cond_nec} is reduced to the conventional uniform mean-unbiasedness:
\[
{\rm{E}}[\hat{\thetavec}-\thetavec;\thetavec]=\zerovec,~\forall \thetavec\in{\mathbb{R}}^M.
\]
Thus, the C-unbiasedness, which is the Lehmann-unbiasedness under constraints,  is the generalization of the conventional uniform mean-unbiasedness for the constrained setting.
\end{example}

\begin{example} 
\color{black} In order to illustrate the C-unbiasedness restriction, which is milder than the conventional mean-unbiasedness, we consider $M=2$, $\Wmat=\Imat_2$, and the circular constraint with radius, $R$, given by $f(\thetavec)=\theta_1^2 + \theta_2^2 - R^2=0$. In this case, $\Umat(\thetavec)=[\theta_2,-\theta_1]^T$. Then, as illustrated in Fig. \ref{mean_bias}, the mean-unbiasedness, or equivalently, the zero mean-bias requirement, 
\begin{equation*}
{\rm{E}}[\hat{\thetavec}-\thetavec;\thetavec]=\zerovec, \forall \thetavec\in\Theta_\fvec, 
\end{equation*}
implies that at any point, $\thetavec\in\Theta_\fvec$, the components of the bias vector should be zero both in the feasible direction of the circular constraint and in the orthogonal direction. In contrast, as illustrated in Fig. \ref{Lehmann_bias}, the uniform C-unbiasedness, or equivalently, the zero C-bias requirement, 
\begin{equation*}
\theta_2{\rm{E}}[\hat{\theta}_1 - \theta_1] -\theta_1{\rm{E}}[\hat{\theta}_2 - \theta_2] = 0, \forall \thetavec\in\Theta_\fvec, 
\end{equation*}
is less restrictive and implies that at any point, $\thetavec\in\Theta_\fvec$, the component of the bias vector should be zero in the feasible direction of the circular constraint but may be nonzero in the orthogonal direction.
\end{example} 
\begin{figure}[hbt]
 \centering
 \begin{subfigure}[b]{0.3\textwidth}
 \centering
 \includegraphics[width=\textwidth]{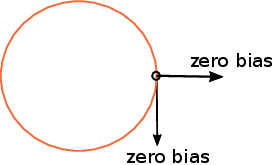}
 \caption{Zero mean-bias requirement for circular constraint.}
 \label{mean_bias}
 \end{subfigure}
 \hfill
 \begin{subfigure}[b]{0.3\textwidth}
 \centering
 \includegraphics[width=\textwidth]{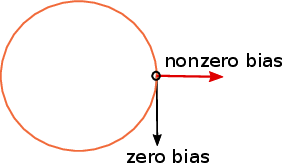}
 \caption{Zero C-bias requirement for circular constraint.}
 \label{Lehmann_bias}
\end{subfigure}
\caption{Different bias requirements for circular constraints.}
\label{mean_bias_Lehmann_bias}
\end{figure}

\subsection{LU-CCRB}
\label{LU_CCRB_subsec}
In this subsection, we present the LU-CCRB on the WMSE from our previous work \cite{JOURNAL_CONSTRAINTS_EYAL}.
\color{black}
By using the C-unbiasedness definition, we developed the LU-CCRB on the WMSE of locally C-unbiased estimators \cite{JOURNAL_CONSTRAINTS_EYAL}:
\be\label{LU_CCRB_W_WMSE}
{\text{WMSE}}_{\hat{\thetavecsmall}}(\thetavec)\geq B_{\text{LU-CCRB}}(\thetavec,\Wmat),
\ee
where
\be \label{total_CR_bound_CS}
\begin{split}
&B_{\text{LU-CCRB}}(\thetavec,\Wmat)\define {\rm{vec}}^T\left(\Umat^T(\thetavec)\Wmat\Umat(\thetavec)\right)\\
&\times\left((\Umat^T(\thetavec)\Wmat\Umat(\thetavec))\otimes(\Umat^T(\thetavec)\Jmat(\thetavec)\Umat(\thetavec))+\Cmat_{\Umat,\Wmat}(\thetavec)\right)^{\dagger}\\
&\times{\rm{vec}}\left(\Umat^T(\thetavec)\Wmat\Umat(\thetavec)\right),
\end{split}
\ee
\be\label{Gamma_block_define}
\begin{split}
&\Cmat_{\Umat,\Wmat}(\thetavec)\define\begin{bmatrix}
\Umat^T(\thetavec)\Vmat_1^T(\thetavec)\\
\vdots\\
\Umat^T(\thetavec)\Vmat_{M-K}^T(\thetavec)
\end{bmatrix}\\
&\times(\Wmat-\Wmat\Umat(\thetavec)(\Umat^T(\thetavec)\Wmat\Umat(\thetavec))^{\dagger}\Umat^T(\thetavec)\Wmat)\\
&\times\begin{bmatrix}
\Vmat_1(\thetavec)\Umat(\thetavec) & \ldots & \Vmat_{M-K}(\thetavec)\Umat(\thetavec)
\end{bmatrix},
\end{split}
\ee
\be\label{V_define}
\Vmat_m(\thetavec)\define\nabla_{\thetavecsmall}\uvec_m(\thetavec),
\ee
and $\uvec_m(\thetavec)$ is the $m$th column of $\Umat(\thetavec)$, which is assumed to be differentiable, $\forall m=1,\ldots,M-K$. It is shown in \cite{JOURNAL_CONSTRAINTS_EYAL} that for linear constraints and/or in the asymptotic region, the LU-CCRB from \eqref{total_CR_bound_CS} coincides with the corresponding CCRB from \eqref{CCRB_W}.

\section{LU-CBTB} \label{sec:LU-CBTB}
In this section, the LU-CBTB is derived in Subsection \ref{subsec:Derivation of LU-CBTB} and some of its properties are investigated in Subsection \ref{subsec:Properties of LU-CBTB}.

\subsection{Derivation of LU-CBTB} \label{subsec:Derivation of LU-CBTB}
In this section, we derive the LU-CBTB on the WMSE of C-unbiased estimators. For the derivation we define the matrix $\Dmat(\thetavec,\Wmat,\Pimat)\in\mathbb{R}^{(P+1)(M-K)\times (P+1)(M-K)}$, which is a block matrix whose $(m,n)$th block is given by
\beqna\label{Dmat_block_define}
\Dmat^{(m,n)}(\thetavec,\Wmat,\Pimat)\define \Umat^T(\thetavec_m)\Wmat\Umat(\thetavec_n)
\nonumber\hspace{2cm}\\
\times \left\{\begin{array}{ll}
  B_{m,n}(\thetavec,\Pimat)&{\text{if }} m\neq 0 ~\& \& ~n\neq 0
\\
1  & {\text{otherwise }}\end{array}
\right.,
\eeqna
$\forall m,n=0,\ldots,P$, where the elements of $\Bmat(\thetavec,\Pimat)$ are defined in \eqref{Bmat_BTB_elements}. In addition, we define the vector $\tvec(\thetavec,\Wmat,\Pimat)\in\mathbb{R}^{P(M-K)}$, which is composed of $P$ stacked subvectors, where the $m$th subvector is given by
\be\label{tvec_block_define}
\tvec^{(m)}(\thetavec,\Wmat,\Pimat)\define\Umat^T(\thetavec_m)\Wmat(\thetavec_m-\thetavec),
\ee
$\forall m=1,\ldots,P$. For the derivation of the LU-CBTB it is sufficient to require pointwise C-unbiasedness at $\thetavec$, as given by \eqref{unbiased_cond_nec}, and at the test-point vectors $\thetavec_1,\ldots,\thetavec_P\in\Theta_\fvec$, i.e.
\be \label{point_unbiased_cond_nec}
\Umat^T(\thetavec_m)\Wmat{\rm{E}}[\hat{\thetavec}-\thetavec_m;\thetavec_m]=\zerovec,~\forall m=0,\ldots,P,
\ee
which is a milder condition than the uniform C-unbiasedness from \eqref{unbiased_cond_nec}.
It should be noted that in \eqref{point_unbiased_cond_nec} we use the notation 
 $\thetavec_0\define\thetavec$.

In the following theorem, the LU-CBTB is presented. 
\begin{Theorem} \label{T3}
Let $\thetavec_1,\ldots,\thetavec_P\in\Theta_\fvec$ be test-point vectors arranged in a matrix $\Pimat$ from \eqref{test_point_matrix}, and assume that the elements of $\Dmat(\thetavec,\Wmat,\Pimat)$ in \eqref{Dmat_block_define} are finite. Let $\hat{\thetavec}$ be an estimator of $\thetavec$ satisfying \eqref{point_unbiased_cond_nec}, i.e., $\hat{\thetavec}$ is a C-unbiased estimator at $\thetavec$ and at the test-points. Then,
\be \label{total_CR_bound_CS_pre_W}
\begin{split}
&{\text{WMSE}}_{\hat{\thetavecsmall}}(\thetavec)\geq B_{\text{LU-CBTB}}(\thetavec,\Wmat)\\
&\define\underset{\Pimatsmall\in\Theta_\fvec^P}{\sup}\left\{\tvec^T(\thetavec,\Wmat,\Pimat)\left[ \Dmat^{\dagger}(\thetavec,\Wmat,\Pimat)\right]_{LR}\tvec(\thetavec,\Wmat,\Pimat)\right\},
\end{split}
\ee
where 
$[ \Dmat^{\dagger}(\thetavec,\Wmat,\Pimat)]_{LR}$ is the lower-right $P(M-K)\times P(M-K)$ block of $\Dmat^{\dagger}(\thetavec,\Wmat,\Pimat)$.
Equality in \eqref{total_CR_bound_CS_pre_W} is obtained {\em iff} the estimator $\hat{\thetavec}$ satisfies
\be \label{equality_cond_Prop_W}
\begin{split}
&\Wmat^{\frac{1}{2}}(\hat{\thetavec}-\thetavec)
\\&=\Wmat^{\frac{1}{2}}\sum_{m=1}^{P}\bigg(\Umat(\thetavec_m)\frac{f_\xvec(\xvec;\thetavec_m)}{f_\xvec(\xvec;\thetavec)}
\\
&~\times[[\Dmat^{\dagger}(\thetavec,\Wmat,\Pimat)]_{LR}\tvec(\thetavec,\Wmat,\Pimat)]_{((m-1)(M-K)+1):(m(M-K))}\bigg).
\end{split}
\ee
for the optimal choice of test-point vectors.
\end{Theorem}
\begin{proof}
The proof is given in Appendix \ref{App_LU_CBTB_THEOREM}.
\end{proof} 
\noindent
 Theorem \ref{T3} describes a general class of lower bounds, where for each choice of a weighting matrix, $\Wmat$, we obtain a different bound.
In particular, 
the bound on the trace of the MSE matrix can be obtained by substituting 
$\Wmat=\Imat_M$ in  Theorem \ref{T3}.
Similarly, by substituting $\Wmat=\avec \avec^T$, where $\avec\in\mathbb{R}^{M}$ is an arbitrary vector, we can obtain a lower bound on any linear combination of the estimation errors, given by \eqref{WMSE_for_aa}.
It should be noted that the parametric constraints affect the CBTB from \eqref{CMS_W} only through the space of the test-point vectors defined in \eqref{set_def}, while the LU-CBTB from \eqref{total_CR_bound_CS_pre_W} employs the information about the constraints also through the matrix $\Umat(\thetavec)$ from \eqref{one} and \eqref{two}.

The following claim presents an approach to compute $[ \Dmat^{\dagger}(\thetavec,\Wmat,\Pimat)]_{LR}$ that is based on the Schur complement.
\begin{claim}
For invertible matrices $\Dmat(\thetavec,\Wmat,\Pimat)$ and $\Umat^T(\thetavec)\Wmat\Umat(\thetavec)$, the matrix $[ \Dmat^{\dagger}(\thetavec,\Wmat,\Pimat)]_{LR}\in\mathbb{R}^{P(M-K)\times P(M-K)}$ from the LU-CBTB in 
 \eqref{total_CR_bound_CS_pre_W}
can be obtained as the inverse of the  block matrix whose $(m,n)$th block of size $(M-K)\times (M-K)$ is given by
\be\label{Dmat_c_block_define_boddy}
\begin{split}
&\Umat^T(\thetavec_{m})\Wmat\Umat(\thetavec_{n})B_{m,n}\\
&~~~-\Umat^T(\thetavec_{m})\Wmat\Umat(\thetavec)(\Umat^T(\thetavec)\Wmat\Umat(\thetavec))^{-1}\Umat^T(\thetavec)\Wmat\Umat(\thetavec_{n})
\end{split}
\ee
$\forall m,n=1,\ldots,P$.
\end{claim}
\begin{proof}
By using \eqref{Dmat_block_define}, 
and applying blockwise matrix inversion (see e.g. \cite[p. 46]{MATRIX_COOKBOOK}) on the matrix $\Dmat^{(m,n)}(\thetavec,\Wmat,\Pimat)$,
we readily obtain the Schur complement property in \eqref{Dmat_c_block_define_boddy}.
\end{proof}

\color{black}

\subsection{Properties of LU-CBTB} \label{subsec:Properties of LU-CBTB}
In this subsection, we show the relations of the LU-CBTB with the LU-CCRB from \eqref{total_CR_bound_CS} and the CBTB from \eqref{CBTBmat}.

\subsubsection{Relation to LU-CCRB}
In the following, we show that the LU-CBTB is tighter than or equal to the LU-CCRB. The proof is based on the fact that the LU-CCRB can be obtained from the LU-CBTB for a specific choice of test-point vectors, which is not necessarily the choice that maximizes the LU-CBTB.

\begin{proposition} \label{order_relation_prop}
Assume that the following conditions are satisfied:
\renewcommand{\theenumi}{C.\arabic{enumi}} 
\begin{enumerate}
\setcounter{enumi}{0}
\item \label{cond1CRB} 
The LU-CBTB and the LU-CCRB exist. In particular, the derivative of the log-likelihood exists and \eqref{l_define} is well-defined.
\item \label{cond3CRB} 
Expectation and limits w.r.t. test-point vectors of the LU-CBTB can be interchanged.
\end{enumerate} 
\renewcommand{\theenumi}{\arabic{enumi}}
Then,
\be\label{order_relation}
B_{\text{LU-CBTB}}(\thetavec,\Wmat)\geq B_{\text{LU-CCRB}}(\thetavec,\Wmat).
\ee
\end{proposition}
\begin{proof}
The proof is given in Appendix \ref{App_LU_CBTB_LU_CCRB_prop}.
\end{proof}

\subsubsection{Relation to CBTB}
First, we show that in general, the proposed LU-CBTB is lower than or equal to the CBTB from \cite{Hero_constraint}. Then, we show that for the special case of linear constraints the two bounds are equal. 
\begin{proposition} \label{CBTB_order_prop}
Assume that the following conditions are satisfied: 
\renewcommand{\theenumi}{C.\arabic{enumi}} 
\begin{enumerate}
\setcounter{enumi}{3}
\item \label{cond1CBTB} 
The CBTB and the LU-CBTB exist.
\item \label{cond2CBTB} 
${\mathcal{R}}(\Tmat^T(\thetavec,\Pimat))\subseteq{\mathcal{R}}(\Bmat^T(\thetavec,\Pimat)-\onevec \onevec^T)$, for any $\Pimat$ such that $\thetavec,\thetavec_1,\ldots,\thetavec_P\in\Theta_\fvec$.
\end{enumerate} 
\renewcommand{\theenumi}{\arabic{enumi}}
Then,
\be\label{order_CBTBs}
B_{\text{LU-CBTB}}(\thetavec,\Wmat)\leq B_{\text{CBTB}}(\thetavec,\Wmat).
\ee
\end{proposition}
\begin{proof}
The proof is given in Appendix \ref{App_CBTB_order_prop}.
\end{proof}
Next, we consider the special case of linear equality constraints, which are in the form
\be\label{linear_constr_eq}
\fvec(\thetavec)=\Amat\thetavec=\zerovec,
\ee
where $\Amat\in\mathbb{R}^{K\times M}$. In the following proposition, we show that the LU-CBTB from \eqref{total_CR_bound_CS_pre_W} coincides with the conventional CBTB on the WMSE from \eqref{CMS_W} for linear constraints. It should be noted that under linear constraints, the CBTB and CCRB may be different. Moreover, the CBTB and the LU-CBTB may exist in cases where the CCRB and the LU-CCRB do not, as shown in the simulations in Section \ref{sec:Simulations}.
\begin{proposition}\label{linear_constraints_prop}
Assume that the constraints are in the form of \eqref{linear_constr_eq} and that the CBTB and the LU-CBTB exist. Then,
\be\label{linear_constraint_equality}
B_{\text{LU-CBTB}}(\thetavec,\Wmat)=B_{\text{CBTB}}(\thetavec,\Wmat).
\ee
\end{proposition}
\begin{proof}
The proof is given in Appendix \ref{App_linear_constraints_prop}.
\end{proof}
Since the result in \eqref{linear_constraint_equality} holds for any $\Wmat$, by taking $\Wmat=\avec \avec^T$, it can be seen that according to \eqref{WMSE_for_aa}
in this case, we obtain
\begin{align}
\label{WMSE_for_aa_bound}
\avec^T {\rm{E}}\left[(\hat{\thetavec}-\thetavec) (\hat{\thetavec}-\thetavec)^T;\thetavec \right] \avec \geq 
B_{\text{LU-CBTB}}(\thetavec,\avec \avec^T)
\\=B_{\text{CBTB}}(\thetavec,\avec \avec^T)=
\avec^T \Bmat_{\text{CBTB}}(\thetavec)\avec,
\end{align}
where the last equality stems from \eqref{CBTBmat} and
\eqref{CMS_W}.
Thus, in the case of linear constraints, the proposed LU-CBTB can be written in a matrix form, as the conventional CBTB.
However, in the general case, the  LU-CBTB depends on the specific choice of $\Wmat$, and can be interpreted as a lower bound on arbitrary linear combinations of estimation error. This way, we obtain tighter lower bounds that fit the parameter estimation problem at hand.

\subsubsection{Unconstrained parameter estimation} The special case $K=0$ implies an unconstrained estimation problem in which $\Umat(\thetavec)=\Imat_M$ for any $\thetavec$. In this case, if we take the weighting matrix $\Wmat=\Imat_M$,  then the matrix $\Dmat^{(m,n)}(\thetavec,\Wmat,\Pimat)$ from \eqref{Dmat_block_define} is reduced to
\beqna\label{Dmat_block_define)unconstrained}
\Dmat(\thetavec,\Wmat,\Pimat)= \left[\begin{array}{cc} 1 &\onevec^T\\
\onevec & \Bmat(\thetavec,\Pimat) 
\end{array}\right]\otimes \Imat_M.
\eeqna
Thus, under the assumption that 
 $\Dmat(\thetavec,\Wmat,\Pimat)$ is a non-singular matrix, 
by using Kronecker product rules, we obtain
\beqna
\label{D_inv}
\Dmat^{-1}(\thetavec,\Wmat,\Pimat)=\left[\begin{array}{cc} 1 &\onevec^T\\
\onevec & \Bmat(\thetavec,\Pimat) 
\end{array}\right]^{-1} \otimes \Imat_M.
\eeqna
Then, computing the lower-right $P M\times P M$ block of $\Dmat^{-1}(\thetavec,\Wmat,\Pimat)$ from \eqref{D_inv} by using Schur complement properties, results in
\beqna
\label{inv_block}
\left[ \Dmat^{-1}(\thetavec,\Wmat,\Pimat)\right]_{LR}
=\left( \Bmat(\thetavec,\Pimat)-\onevec\onevec^T\right)^{-1} \otimes \Imat_M.
\eeqna
In addition, the vector $\tvec(\thetavec,\Wmat,\Pimat)$, defined in \eqref{tvec_block_define}, is reduced in this case (i.e., where $\Umat(\thetavec)=\Wmat=\Imat_M$) to
\be\label{tvec_block_define_app}
\tvec(\thetavec,\Wmat,\Pimat)={\rm{vec}}(\Tmat(\thetavec,\Pimat)),
\ee
where $\Tmat(\thetavec,\Pimat)$ is defined in \eqref{T_BTB}.
By substituting \eqref{inv_block} and \eqref{tvec_block_define_app} in the LU-CBTB from \eqref{total_CR_bound_CS_pre_W}, we obtain that
\beqna \label{total_CR_bound_CS_pre_W_unconstrained}
B_{\text{LU-CBTB}}(\thetavec,\Wmat)
=\underset{\Pimatsmall\in\Theta_\fvec^P}{\sup}\left\{{\rm{vec}}^T(\Tmat(\thetavec,\Pimat))\right.\hspace{1.25cm}\nonumber\\
\left. \times\left(
\left( \Bmat(\thetavec,\Pimat)-\onevec\onevec^T\right)^{-1} \otimes \Imat_M\right){\rm{vec}}(\Tmat(\thetavec,\Pimat))\right\}
\nonumber\\
=\underset{\Pimatsmall\in\Theta_\fvec^P}{\sup}\left\{
{\rm{Tr}}\left( \Tmat(\thetavec,\Pimat)
\left( \Bmat(\thetavec,\Pimat)-\onevec\onevec^T\right)^{-1}  \Tmat^T(\thetavec,\Pimat)\right) \right\},
\eeqna
where the last equality is obtained by using  \eqref{kronecker_prop} from Appendix \ref{App_linear_constraints_prop}.
It can be seen that in \eqref{total_CR_bound_CS_pre_W_unconstrained}, we obtain the conventional BTB on the trace of the MSE matrix, without constraints.
\color{black}
\section{Example: DOA estimation for constant modulus discrete signal} \label{sec:Simulations}
We begin this section by presenting the model and parametric constraints in Subsection \ref{subsec:Model and constrained setting}. Then, we present the CML estimator and the unbiasedness requirements in Subsection \ref{subsec:CML estimator}, where the corresponding performance bounds are presented in Subsection \ref{performance_bounds_sim_subsec}. The numerical results are presented in Subsection \ref{subsec:Numerical results}. For brevity of this section, we remove the arguments of vector/matrix functions, except for specific derivations/functions where the arguments are necessary.

\subsection{Model and constrained setting}\label{subsec:Model and constrained setting}
In this section, we consider a snapshot of a CM signal propagating through a homogeneous medium towards a uniform circular array (UCA) of $Q$ sensors with radius $r$. 
 The source and the array are assumed to be coplanar, i.e. the source is located at the same plane of the array.
CM signals are commonly encountered in radar and communications with phase modulated signals \cite{VAN_DER_VEEN_CM,LESHEM_CM,STOICA_CM}, where in some cases, the signals are discrete and belong to a finite or countable set \cite{GAMBOA,DELMAS_DOA,DELMAS_DOA_SIGNALS,KRUMMENAUER,MEDINA_INTEGER,LESHEM_DISCRETE,MANIOUDAKIS,YANG_DISCRETE}. 
The complex baseband array output is modeled as follows (see, e.g., \cite{LESHEM_CM,STOICA_CM})
\be \label{model}
\xvec=\alpha e^{j\phi}\avec(\nu)+\nvec,
\ee
where $\alpha>0$ is the amplitude of the received signal and $\nu
\in\mathbb{C}$, such that $|\nu|=1$.
The steering vector, $\avec(\nu)$, is obtained from a source impinging from direction $\angle{\nu}\in[-\pi,\pi)$ and satisfies $\|\avec(\nu)\|^{2}=Q$. In particular, the $q$th element of $\avec(\nu)$ is 
\be
\label{a_q_def}
a_{q}(\nu)=e^{j \frac{2\pi r}{\lambda}{\text{Re}}\left\{\nu e^{-j \frac{2\pi (q-1)}{Q}}\right\}},~~~\forall q=1,\ldots,Q,
\ee
where 
$\lambda$ is the signal wavelength. The CM signal $e^{j\phi}$ is known to belong to a finite set, $\mathcal{S}$. The noise, $\nvec$, is a spatially and temporally white complex Gaussian random vector with known covariance matrix $\sigma^2 \Imat_Q$. The SNR is defined as ${\text{SNR}}\define\frac{\alpha^2}{\sigma^2}$.\\
\indent 
 The unknown parameter vector in this case is 
 \be
 \label{theta_def_sim}
 \thetavec=[\theta_1,\theta_2,\theta_3,\theta_4]^T=[{\text{Re}}\{\nu\},{\text{Im}}\{\nu\},\phi,\alpha]^T
 \ee under the constraint
\be\label{f_example}
f(\thetavec)=\theta_1^2+\theta_2^2-1=0.
\ee
Norm constraints in the form of 
$\sum_{k=1}^K \theta_k^2 -\rho^2=0$, with  $1\leq K \leq M$ and $\rho>0$,
are commonly used in various signal processing problems  (see, e.g. \cite{Moore_scoring,Golub,CHEN_REGULAR,normC}).
For norm constraints, reparameterization of the original problem results in a periodic distribution, for which there is no
uniformly mean-unbiased estimator of the parameter vector \cite{PHASE_KAY,CYCLIC,TODROS_WINNIK}.
\\
\indent
Taking the gradient of \eqref{f_example}, we obtain
\be\label{F_example}
\Fmat(\thetavec)=2\cdot [\theta_1,\theta_2,0,0]. 
\ee
By substituting \eqref{F_example} in \eqref{one} and \eqref{two}, it can be shown that the matrix
\be\label{U_example}
\Umat(\thetavec)=\begin{bmatrix}
\theta_2 & 0 & 0\\
-\theta_1 & 0 & 0\\
0 & 1 & 0\\
0 & 0 & 1
\end{bmatrix}
\ee
is an orthonormal complement matrix for the considered problem.
\subsection{Unbiasedness and the CML estimator}
\label{subsec:CML estimator}
In this example, we are interested in the estimation of $\theta_1$ and $\theta_2$ that determine the DOA of the signal. The CM signal phase and amplitude, $\phi$ and $\alpha$, respectively, are considered as nuisance parameters. Thus, according to the definition of the unknown parameter vector $\thetavec$ in \eqref{theta_def_sim}, we choose the weighting matrix
\be\label{W_choice}
\Wmat={\rm{diag}}([1,1,0,0]^T).
\ee
It is shown in \cite{SOMEKH_LESHEM} that under a norm constraint, as in \eqref{f_example}, there is no uniformly mean-unbiased estimator of the parameter vector. However, it will be shown later in this section that a C-unbiased estimator of the parameter $\nu$ can be found, and a corresponding informative performance bound can be derived.

In this case, by inserting \eqref{U_example} and (\ref{W_choice}) into the C-unbiasedness condition from \eqref{unbiased_cond_nec}, we obtain the requirement
\be\label{C_bias_examp}
\begin{split}
\big[{\rm{E}}[\theta_2\hat{\theta}_1-\theta_1\hat{\theta}_2;\thetavec],0,0\big]^T=\zerovec,~\forall\thetavec\in\Theta_\fvec.
\end{split}
\ee
Thus,  the Lehmann requirement for uniform C-unbiasedness in this case is
\be
\label{unbiased_cond_nec_sim}
{\rm{E}}[\hat{\theta}_1;\thetavec]\theta_2-{\rm{E}}[\hat{\theta}_2;\thetavec]\theta_1=0
,~\forall\thetavec\in\Theta_\fvec,
\ee
where we use the fact that $\theta_1$ and $\theta_2$ are deterministic. 
It can be seen that the condition stated in \eqref{unbiased_cond_nec_sim} is less restrictive compared to the mean-unbiasedness requirement, which is ${\rm{E}}[\hat{\theta}_i;\thetavec]=\theta_i$, $\forall i=1,2,3,4$.
Moreover, it is worth mentioning
that, unlike mean-unbiasedness, there always exists at least one C-unbiased estimator. A 
C-unbiased estimator can be obtained by setting $\hat{\theta}_1=\cos\beta$ and $\hat{\theta}_2=\sin\beta$, where $\beta$ is a random variable uniformly distributed in the range $[-\pi,\pi)$ and independent of $\xvec$. In this scenario, both ${\rm{E}}[\hat{\theta}_1;\thetavec]$ and ${\rm{E}}[\hat{\theta}_2;\thetavec]$ are equal to zero. Thus, it is a C-unbiased estimator that satisfies \eqref{unbiased_cond_nec_sim}. Although the performance of this estimator in terms of WMSE is poor, it demonstrates the existence of C-unbiased estimators for this problem, whereas the set of mean-unbiased estimators is empty \cite{SOMEKH_LESHEM}.

Taking into account the discrete nature of $e^{j\phi}$, or equivalently $\phi$, the CML estimators of $\nu$ and $\phi$ are given by \cite{STOICA_CM}: $\hat\nu=e^{j \angle{\hat\nu}}$ and $\hat\phi$, respectively, where
\be\label{CML_phases}
(\angle{\hat\nu},\hat\phi)=\arg\underset{\angle{\nu}\in[-\pi,\pi),\phi\in\angle{\mathcal{S}}}{\max}{\text{Re}}\{\avec^{H}(\nu)\xvec e^{-j\phi}\},
\ee
in which $\angle{\mathcal{S}}$ is the set of phases of the complex numbers in $\mathcal{S}$. The CML estimator of $\alpha$ is given by 
\[\hat{\alpha}=\frac{1}{Q}{\text{Re}}\{\avec^{H}(\hat{\nu})\xvec e^{-j\hat\phi}\}.\]\\
\indent
By using the CML  from \eqref{CML_phases}, 
it can be verified that, in this case,  the 
C-unbiasedness condition from \eqref{unbiased_cond_nec_sim} is reduced to the requirement
\beqna
\label{unbiased_cond_nec_sim2}
{\rm{E}}[\cos \angle{\hat\nu};\thetavec]\sin(\angle{\nu})-{\rm{E}}[\sin\angle{\hat\nu};\thetavec]\cos(\angle{\nu})
\nonumber\\=
{\rm{E}}[
\sin(\angle{\nu}-\angle{\hat\nu});\thetavec]=0.
\eeqna
The uniformly C-unbiasedness condition in \eqref{unbiased_cond_nec_sim2} coincides with
the uniformly periodic unbiasedness definition of DOA estimation in previous
works (e.g. \cite{DIRECTIONAL} pp. 118–119, \cite{CYCLIC}).
\color{black}
In order to examine the mean-unbiasedness and C-unbiasedness of the CML estimator, we numerically evaluate its mean-bias and C-bias from \eqref{C_bias_examp} in Subsection \ref{subsec:Numerical results}, where the mean-bias is evaluated only for the estimation of $\nu$.
\subsection{Performance bounds}
\label{performance_bounds_sim_subsec}
Since $\phi$ is discrete, the derivative of the likelihood function w.r.t. it is not defined. Accordingly, the CCRB from \eqref{CCRB_W} and the LU-CCRB from \eqref{total_CR_bound_CS}, which require existence of the derivatives of the likelihood function w.r.t. all the unknown parameters, do not exist in this case. In contrast, the CBTB from \eqref{CMS_W} and the LU-CBTB from \eqref{total_CR_bound_CS_pre_W} can be derived. For the derivation, we need to choose test-point vectors that satisfy the constraints in \eqref{f_example}. Thus, the test-point vectors can be written as
\be\label{theta_tp_structure}
\thetavec_p=[\cos\omega_p,\sin\omega_p,\phi_p,\alpha_p]^T,
\ee
$\forall p=1,\ldots,P$, where $\omega_p\in[-\pi,\pi)$, $\phi_{p}\in\angle{\mathcal{S}}$, and $\alpha_p>0$. For both the CBTB and the LU-CBTB, the discrete nature of $\phi$ is manifested in the test-point vectors from \eqref{theta_tp_structure}.\\
\indent
One of the main advantages of Cram$\acute{\text{e}}$r-Rao-type bounds is that in many cases they can be derived in closed-form without numerical matrix inversion. For BTBs it is usually more difficult to avoid this numerical matrix inversion. However, we show that for a specific choice of the test-point vectors, we can obtain closed-form expressions for the CBTB and LU-CBTB in this case. 
We set $P=3$ and choose the following test-point vectors:
\be\label{tp_nu}
\thetavec_1=[\cos(\angle{\nu}+h_\nu),\sin(\angle{\nu}+h_\nu),\phi,\alpha]^T,
\ee
\be\label{tp_s}
\thetavec_2=[\cos(\angle{\nu}),\sin(\angle{\nu}),\phi+h_\phi,\alpha]^T,
\ee
and
\be\label{tp_alpha}
\thetavec_3=[\cos(\angle{\nu}),\sin(\angle{\nu}),\phi,\alpha+h_\alpha]^T,
\ee
where $h_\nu\in[-\pi,\pi)$, $(\phi+h_\phi)\in\mathcal{S}$, and $(\alpha+h_\alpha)>0$, s.t. the test-point vectors satisfy the constraint from \eqref{f_example} and the discrete nature of $\phi$.

\begin{claim}
\label{claim_sim}
    For the considered model and the test-point vectors given in \eqref{tp_nu}-\eqref{tp_alpha}, the CBTB and LU-CBTB from
  \eqref{CMS_W} and \eqref{total_CR_bound_CS_pre_W}, respectively
    are given by
    \be \label{CBTB_DOA}
\begin{split}
B_{\text{CBTB}}=\underset{h_\nu,h_\phi,h_\alpha}{\sup}\frac{2-2\cos h_\nu}{\bar{B}_{2,2}+1-G_\Bmat(h_\nu,h_\phi,h_\alpha)}
\end{split}
\ee
and
    \be \label{LU_CBTB_DOA}
\begin{split}
B_{\text{LU-CBTB}}=\underset{h_\nu ,h_\phi,h_\alpha}{\sup}\frac{\sin^2{h_\nu}}{\bar{B}_{2,2}+1-G_\Bmat(h_\nu,h_\phi,h_\alpha)\cos^2{h_\nu}},
\end{split}
\ee
where the supremum in both bounds is taken over $h_\nu\in[-\pi,\pi)$, $(\phi+h_\phi)\in\mathcal{S}$, $(\alpha+h_\alpha)>0$,  
\beqna\label{G_B}
G_\Bmat(h_\nu,h_\phi,h_\alpha)\define 1 + \frac{1}{\bar{B}_{3,3}\bar{B}_{4,4} - \bar{B}_{3,4}^2}\hspace{2.25cm}\nonumber\\
\times\big( \bar{B}_{4,4} \bar{B}_{2,3}^2 -2\bar{B}_{2,3}\bar{B}_{2,4}\bar{B}_{3,4} + \bar{B}_{3,3}\bar{B}_{2,4}^2\big).
\eeqna
in which  $\bar{B}_{m,n}$, $m,n=2,3,4$ are given in Appendix \ref{new_app}, in \eqref{Bmat_22}- \eqref{Bmat_44}. 
\end{claim}
\begin{proof}
The proof is given in Appendix \ref{new_app}.
\end{proof}

 By comparing \eqref{CBTB_DOA} and \eqref{LU_CBTB_DOA} and using the facts that $G_\Bmat(h_\nu,h_\phi,h_\alpha)$ is non-negative and $\sin^2{h_\nu}=(1-\cos{h_\nu})(1+\cos{h_\nu})\leq 2-2\cos{h_\nu}$, it can be observed that
 \[B_{\text{LU-CBTB}}\leq B_{\text{CBTB}}, \]
 which is in accordance with Proposition \ref{CBTB_order_prop}. It should be emphasized that this inequality results from the different unbiasedness restrictions imposed by the bounds and not due to lack of tightness of the LU-CBTB compared to the CBTB.

\subsection{Numerical results}\label{subsec:Numerical results}
In the following simulations, we assume that $e^{j\phi}$ is a quadrature phase-shift keying (QPSK) modulated signal \cite{DELMAS_DOA,DELMAS_DOA_SIGNALS}, i.e. $\phi\in\{\frac{\pi}{4},\frac{3\pi}{4},-\frac{3\pi}{4},-\frac{\pi}{4}\}$. We evaluate the LU-CBTB and CBTB from \eqref{LU_CBTB_DOA} and \eqref{CBTB_DOA}, respectively, and compare them to the WMSE of the CML estimator of $\thetavec$ with $\Wmat$ from \eqref{W_choice}. The mean-unbiasedness and C-unbiasedness of the CML estimator for DOA estimation are investigated as well, in accordance with the choice of $\Wmat$ from \eqref{W_choice}. The performance of the CML estimator is evaluated using 10,000 Monte-Carlo trials. To reduce the complexity of the maximization in \eqref{LU_CBTB_DOA}-\eqref{CBTB_DOA}, we set $h_\alpha=10^{-5}$ and only two-dimensional maximization is performed w.r.t. $h_\phi\in\{-\pi,-\frac{\pi}{2},\frac{\pi}{2}\}$ and $h_\nu\in[-\pi,\pi)$. The grid for $h_\nu$ contains $64$ equally spaced points in the interval $[-\pi,\pi)$. Consequently, the implemented bounds require $64\times 3$ two-dimensional grid search w.r.t. $(h_\nu,h_\phi)$. We set $Q=4$, $\sigma^2=0.5$, and $\phi=\frac{\pi}{4}$.
\\
\subsubsection{Case A}
In this case, we set $\frac{2\pi r}{\lambda}=0.5$.
In Fig. \ref{CML_bias_alpha}, we evaluate the norms of the bias and the C-bias of the CML estimator versus SNR for $\angle\nu=0.9\pi$. The C-bias norm is approximately zero for all the considered SNR values implying C-unbiasedness of the CML estimator. In contrast, the considered bias norm for estimation of $\nu$ is nonzero for ${\text{SNR}}<20~{\text{dB}}$. In Fig. \ref{CML_MSE_alpha}, we evaluate the WMSE of the CML estimator, the CBTB, and the LU-CBTB versus SNR. It can be seen that the LU-CBTB is a lower bound on the CML WMSE for all the considered SNR values, while the CBTB is not a lower bound for ${\text{SNR}}<0~{\text{dB}}$. This is due to the mean-bias of the CML estimator. In addition, in terms of capturing the CML error behavior, which is prominent for reliable performance bounds, the LU-CBTB significantly outperforms the CBTB.\\
\indent
In Fig. \ref{CML_bias_nu}, we evaluate the norms of the bias and the C-bias of the CML estimator versus $\angle\nu\in[-\pi,\pi)$ for $\alpha=0.16, \sigma^2=0.5$, or equivalently, ${\text{SNR}}\approx -13~{\text{dB}}$. It can be seen that the CML C-bias norm is approximately zero, while its bias norm is significantly higher than zero for all the considered values of $\angle\nu$. These results imply that the CML estimator is C-unbiased, as required by the LU-CBTB, but not mean-unbiased, as required by the CBTB.\\
\indent
In Fig. \ref{CML_MSE_nu}, we evaluate the WMSE of the CML estimator, the CBTB, and the LU-CBTB versus $\angle\nu\in[-\pi,\pi)$. It can be seen that 
in this case, the CML WMSE and the bounds remain nearly constant for any $\angle\nu$. In addition,
the LU-CBTB is a lower bound on the CML WMSE, while the CBTB is not. Moreover, the CBTB is significantly higher than the CML WMSE, demonstrating its inappropriateness in this case. Overall, the results in Figs. \ref{CML_bias_alpha}-\ref{CML_MSE_nu} show that unlike the LU-CBTB the CBTB is not reliable for performance analysis of the CML estimator in the considered example.\\
\begin{figure}[h!]
\centering\includegraphics[width=7cm]{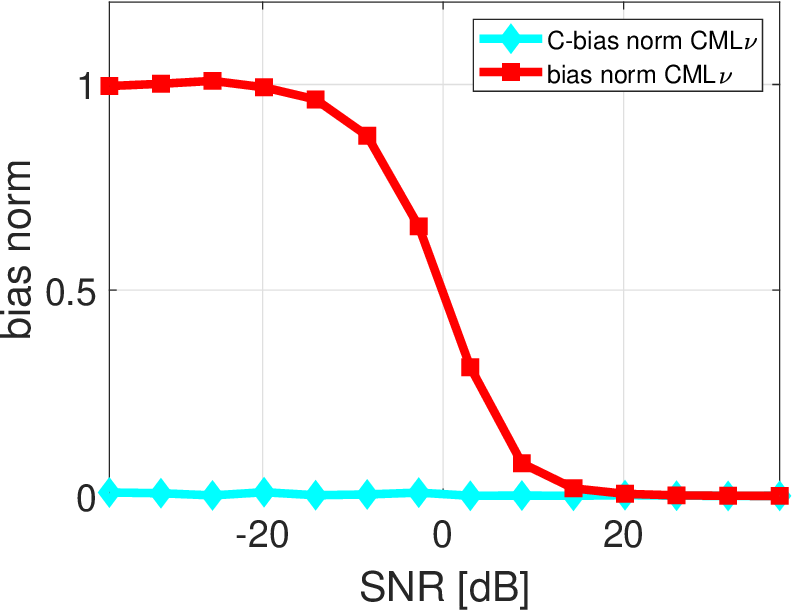}
\caption{The norms of the bias and the C-bias of the CML estimator versus SNR for $\frac{2\pi r}{\lambda}=0.5$.
}\label{CML_bias_alpha}
\end{figure}

\begin{figure}[h!]
\centering\includegraphics[width=7cm]{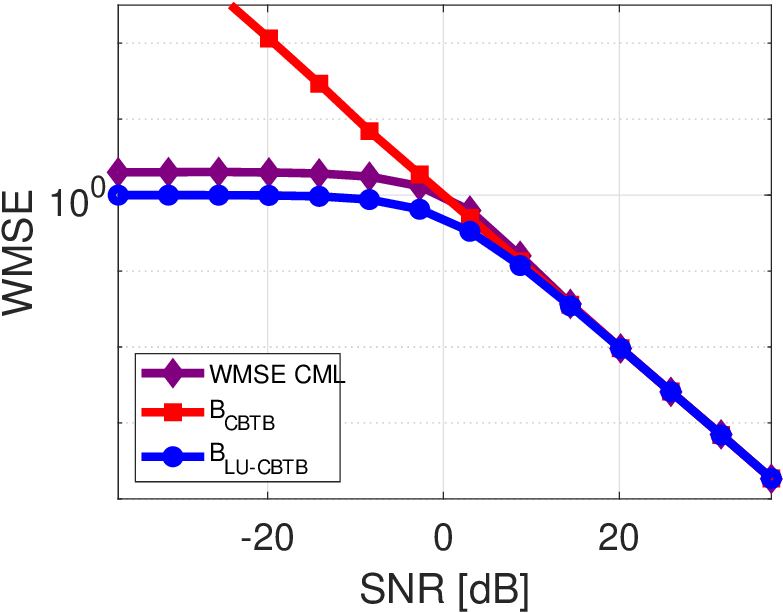}
\caption{The WMSE of the CML estimator, CBTB, and LU-CBTB versus SNR for $\frac{2\pi r}{\lambda}=0.5$.
}\label{CML_MSE_alpha}
\end{figure}

\begin{figure}[h!]
\centering\includegraphics[width=7cm]{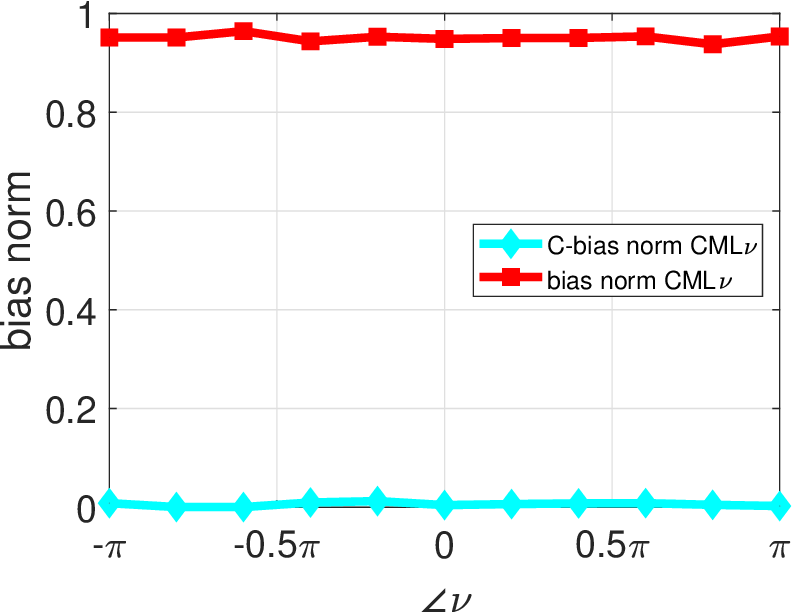}
\caption{The norms of the bias and the C-bias of the CML estimator versus $\angle\nu$.
}\label{CML_bias_nu}
\end{figure}

\begin{figure}[h!]
\centering\includegraphics[width=7cm]{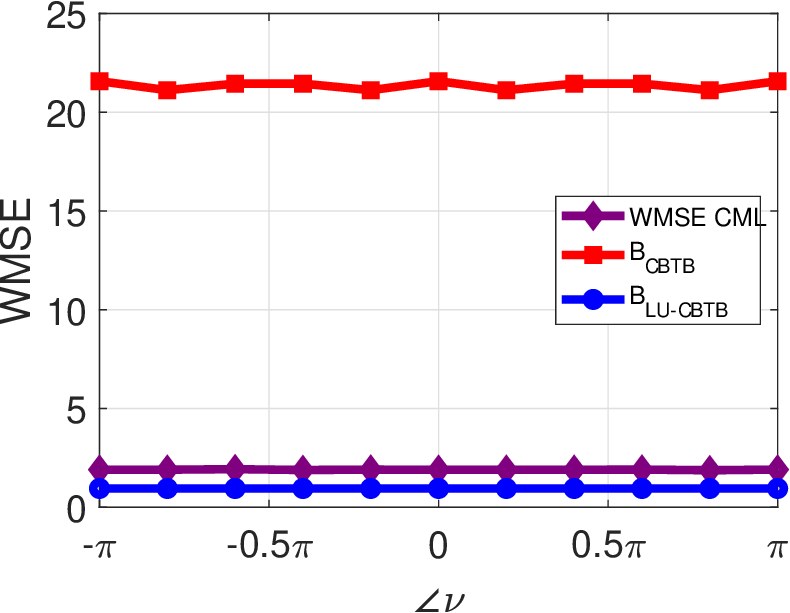}
\caption{The WMSE of the CML estimator, CBTB, and LU-CBTB versus $\angle\nu$.
}\label{CML_MSE_nu}
\end{figure}

\subsubsection{Case B}
In this case, we set $\frac{2\pi r}{\lambda}=3.2$. For UCA with greater radius, the ambiguity in the DOA estimation is expected to be higher, and the threshold SNR of the CML estimator increases \cite{RIFE,TTB1,Richmond_2005}. \color{black} ``Large-error'' lower bounds should take such scenario into account \cite{TTB1,TABRIKIAN_OCEAN}. 
To examine the CML threshold SNR scenario, we reevaluate the performance of the CML estimator versus SNR and compare it to the LU-CBTB from \eqref{LU_CBTB_DOA} using a higher value of $\frac{2\pi r}{\lambda}$, in particular $\frac{2\pi r}{\lambda}=3.2$ rather than $\frac{2\pi r}{\lambda}=0.5$ that was used in Figs. \ref{CML_bias_alpha}-\ref{CML_MSE_alpha}. In addition, we set $\angle\nu=0.9\pi$.\\
\indent
For a higher value of $\frac{2\pi r}{\lambda}$, it is expected that the DOA test-point grid density will have an effect on the tightness of ``large-error'' bounds (see e.g. \cite{TTB1,BZB_EYAL}).
In this region, there are higher sidelobes in the likelihood function and the test-point vectors are expected to be around the sidelobes or ambiguous peaks \cite{TABRIKIAN_OCEAN}. 
However, a denser test-point grid increases the computational complexity of the LU-CBTB. In the following, we are interested in inspecting this tradeoff between computational complexity and tightness for the LU-CBTB. Thus, except for the previously described LU-CBTB that requires $64\times 3$ two-dimensional grid search w.r.t. $(h_\nu,h_\phi)$, we also implement a low-complexity LU-CBTB with $h_\nu=10^{-5}$, i.e. no grid search for $h_\nu$. This low-complexity LU-CBTB requires only a one-dimensional grid search w.r.t. $h_\phi\in\{-\pi,-\frac{\pi}{2},\frac{\pi}{2}\}$.\\
\indent 
In Fig. \ref{CML_bias_alpha_thresh}, we evaluate the norms of the bias and the C-bias of the CML estimator versus SNR. The C-bias norm is approximately zero for ${\text{SNR}}>20~{\text{dB}}$ and ${\text{SNR}}<-20~{\text{dB}}$. The bias norm is higher than zero for ${\text{SNR}}<20~{\text{dB}}$. For $-20~{\text{dB}}\leq{\text{SNR}}\leq 20~{\text{dB}}$ there is a nonzero C-bias norm but significantly lower than the corresponding bias norm.\\
\indent 
Finally, in Fig. \ref{CML_MSE_alpha_thresh}, we evaluate the WMSE of the CML estimator, the LU-CBTB, and the low-complexity LU-CBTB versus SNR. It can be seen that the LU-CBTBs are lower bounds on the CML WMSE for all the considered SNR values. Moreover, for ${\text{SNR}}>25~{\text{dB}}$ the bounds coincide with the WMSE of the CML estimator. The LU-CBTB with the two-dimensional grid search w.r.t. $(h_\nu,h_\phi)$ is tighter than the low-complexity LU-CBTB for ${\text{SNR}}<20~{\text{dB}}$.

\begin{figure}[h!]
\centering\includegraphics[width=7cm]{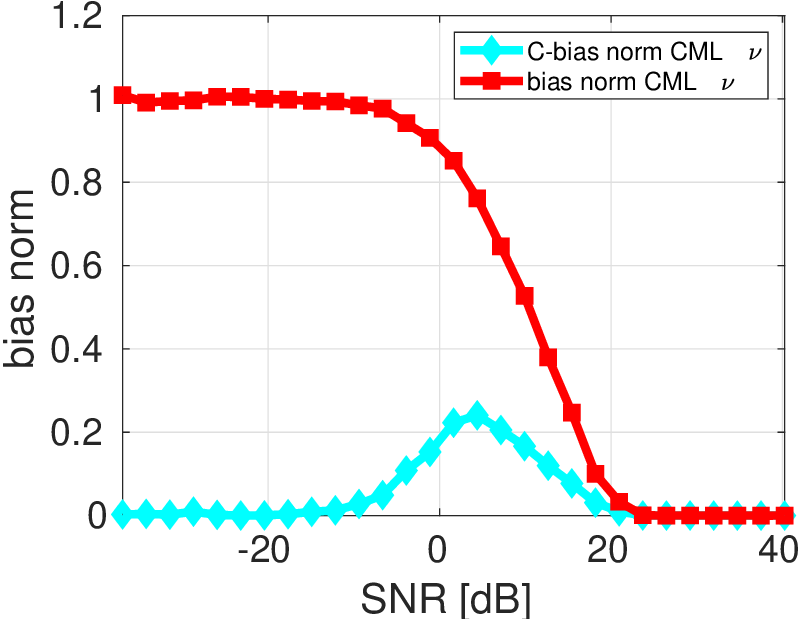}
\caption{The norms of the bias and the C-bias of the CML estimator versus SNR for $\frac{2\pi r}{\lambda}=3.2$.
}\label{CML_bias_alpha_thresh}
\end{figure}

\begin{figure}[h!]
\centering\includegraphics[width=7cm]{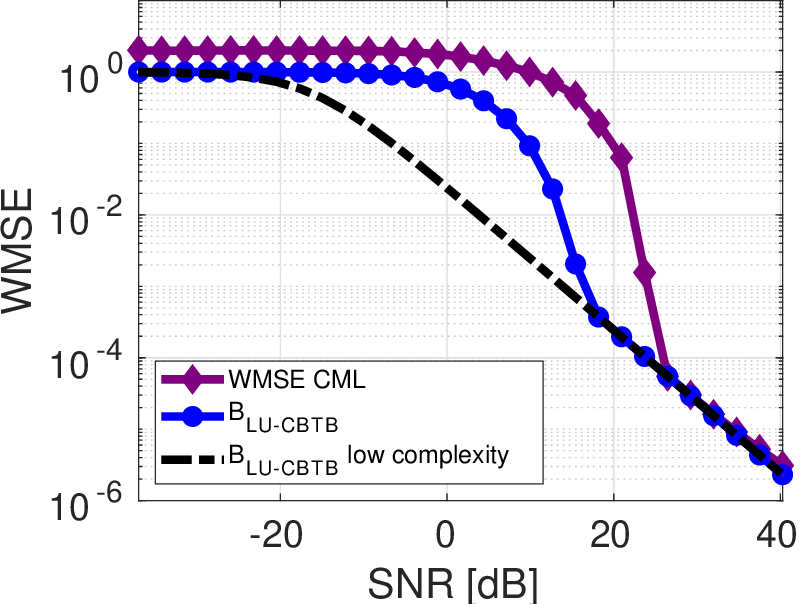}
\caption{The WMSE of the CML estimator, LU-CBTB, and low-complexity LU-CBTB versus SNR for $\frac{2\pi r}{\lambda}=3.2$.
}\label{CML_MSE_alpha_thresh}
\end{figure}

\section{Conclusion} \label{sec:Conclusion}
In this paper, we derived the LU-CBTB for constrained parameter estimation. This bound is a constrained Barankin-type bound that requires Lehmann-unbiasedness, rather than the conventional but restrictive mean-unbiasedness.
The LU-CBTB is a lower bound on the WMSE, and thus, depends on the weighting matrix in the general case. 
It was shown that the LU-CBTB is tighter than or equal to the recently proposed LU-CCRB. The proposed bound coincides with the conventional CBTB for linear constraints, while for nonlinear constraints it may be a useful performance benchmark in cases where the CBTB is not.
Generally, the LU-CBTB is lower than the CBTB, not because it is a loose performance bound but because the CBTB is not a lower bound for common estimators in many problems. In the simulations, we have considered DOA estimation of an unknown QPSK signal. In this example, Cram$\acute{\text{e}}$r–Rao-type bounds do not exist, while the LU-CBTB and the CBTB can be used. It was clearly shown that the LU-CBTB is more appropriate than the CBTB for performance analysis of the CML estimator in terms of providing reliable performance benchmarks and characterizing the CML error behavior. Topics for future work include Lehmann-unbiasedness for misspecified parameter estimation \cite{THANH_MISSP,CCRB_misspecified}, random parametric constraints \cite{PREVOST_CHAUMETTE_RANDOM}, and Bayesian parameter estimation under various parametric structures \cite{JOURNAL_EYAL,ROUTTENBERG_BAY_PER,FAUSS_KULLBACK}.

\appendices

\section{Proof of Theorem \ref{T3}}\label{App_LU_CBTB_THEOREM}
For brevity, in this appendix and in the following appendices, we remove the arguments of vector/matrix functions, except for specific derivations/functions where the arguments are necessary. By using Cauchy-Schwarz inequality with the random vectors $\etavec,\xivec\in\mathbb{R}^M$, one obtains
\be\label{CS_preliminary}
{\rm{E}}[\etavec^T\etavec]{\rm{E}}[\xivec^T\xivec]\geq{\rm{E}}^2[\xivec^T\etavec].
\ee
Let $\lambdavec_m\in\mathbb{R}^{M-K},~m=0,\ldots,P$, be arbitrary vectors. By substituting 
\be
\label{eta_term}
\etavec=\Wmat^{\frac{1}{2}}(\hat{\thetavec}-\thetavec)
\ee
and 
\be
\label{xi_term}
\xivec=\Wmat^{\frac{1}{2}}\sum_{m=0}^{P}\Umat(\thetavec_m)\lambdavec_m\frac{f_\xvec(\xvec;\thetavec_m)}{f_\xvec(\xvec;\thetavec)}
\ee 
in \eqref{CS_preliminary}, where $\thetavec_0=\thetavec$, and using the linearity of the expectation operator, one obtains
\be \label{CS_start_I}
\begin{split}
&{\rm{E}}\left[(\hat{\thetavec}-\thetavec)^T\Wmat(\hat{\thetavec}-\thetavec);\thetavec\right]\sum_{m=0}^{P}\sum_{n=0}^{P}\lambdavec_m^T\Umat^T(\thetavec_m)\Wmat\Umat(\thetavec_n)\lambdavec_n\\
&\times{\rm{E}}\left[\frac{f_\xvec(\xvec;\thetavec_m)}{f_\xvec(\xvec;\thetavec)}\frac{f_\xvec(\xvec;\thetavec_n)}{f_\xvec(\xvec;\thetavec)};\thetavec\right]\\
&\geq\bigg(\sum_{m=0}^{P}\lambdavec_m^T\Umat^T(\thetavec_m)\Wmat{\rm{E}}\left[(\hat{\thetavec}-\thetavec)\frac{f_\xvec(\xvec;\thetavec_m)}{f_\xvec(\xvec;\thetavec)};\thetavec\right]\bigg)^2.
\end{split}
\ee
By substituting \eqref{WMSE} and \eqref{Bmat_BTB_elements} in \eqref{CS_start_I}, one obtains
\be \label{CS_start}
\begin{split}
&{\text{WMSE}}_{\hat{\thetavecsmall}}(\thetavec) \sum_{m=0}^{P}\sum_{n=0}^{P}\lambdavec_m^T\Umat^T(\thetavec_m)\Wmat\Umat(\thetavec_n)\lambdavec_n B_{m,n}\\
&\geq\bigg(\sum_{m=0}^{P}\lambdavec_m^T\Umat^T(\thetavec_m)\Wmat{\rm{E}}\left[(\hat{\thetavec}-\thetavec)\frac{f_\xvec(\xvec;\thetavec_m)}{f_\xvec(\xvec;\thetavec)};\thetavec\right]\bigg)^2,
\end{split}
\ee
where for the sake of simplicity of presentation, in this appendix we use the notation
$B_{0,0}=1$.
Using \eqref{Dmat_block_define} and the augmented vector
\be\label{Su_cvec_define_W}
\lambdavec\define[\lambdavec_0^T,\ldots,\lambdavec_{P}^T]^T,
\ee
we can write the right term in the left hand side (l.h.s.) of \eqref{CS_start} as
\be\label{CS_step_2_Dmat}
\sum_{m=0}^{P}\sum_{n=0}^{P}\lambdavec_m^T\Umat^T(\thetavec_m)\Wmat\Umat(\thetavec_n)\lambdavec_n B_{m,n}=\lambdavec^T\Dmat\lambdavec.
\ee 
Substituting \eqref{CS_step_2_Dmat} in \eqref{CS_start} yields
\be \label{CS_step_1}
\begin{split}
&{\text{WMSE}}_{\hat{\thetavecsmall}}(\thetavec)\cdot(\lambdavec^T\Dmat\lambdavec)\\
&\geq\left(\sum_{m=0}^{P}\lambdavec_m^T\Umat^T(\thetavec_m)\Wmat{\rm{E}}\left[(\hat{\thetavec}-\thetavec)\frac{f_\xvec(\xvec;\thetavec_m)}{f_\xvec(\xvec;\thetavec)};\thetavec\right]\right)^2.
\end{split}
\ee
In addition, by using the C-unbiasedness condition at $\thetavec_0=\thetavec$, we obtain
\beqna
\label{zero_theta}
\Umat^T(\thetavec_0)\Wmat{\rm{E}}\left[(\hat{\thetavec}-\thetavec)\frac{f_\xvec(\xvec;\thetavec_0)}{f_\xvec(\xvec;\thetavec)};\thetavec\right]\hspace{2.5cm}\nonumber\\
=\Umat^T(\thetavec)\Wmat{\rm{E}}\left[\hat{\thetavec}-\thetavec;\thetavec\right]=\zerovec.
\eeqna
Now, by using $\hat{\thetavec}-\thetavec=\hat{\thetavec}-\thetavec_m+\thetavec_m-\thetavec$, it can be verified that
\be\label{factorization}
\begin{split}
&\Umat^T(\thetavec_m)\Wmat{\rm{E}}\left[(\hat{\thetavec}-\thetavec)\frac{f_\xvec(\xvec;\thetavec_m)}{f_\xvec(\xvec;\thetavec)};\thetavec\right]\\
&=\Umat^T(\thetavec_m)\Wmat{\rm{E}}\left[(\hat{\thetavec}-\thetavec_m)\frac{f_\xvec(\xvec;\thetavec_m)}{f_\xvec(\xvec;\thetavec)};\thetavec\right]\\
&~~~+\Umat^T(\thetavec_m)\Wmat(\thetavec_m-\thetavec){\rm{E}}\left[\frac{f_\xvec(\xvec;\thetavec_m)}{f_\xvec(\xvec;\thetavec)};\thetavec\right]\\
&=\Umat^T(\thetavec_m)\Wmat{\rm{E}}\left[\hat{\thetavec}-\thetavec_m;\thetavec_m\right]+\tvec^{(m)}\\
&=\tvec^{(m)},~\forall m=1,\ldots,P,
\end{split}
\ee
where the second equality is obtained by substituting \eqref{tvec_block_define} and using the fact that ${\rm{E}}\left[\frac{f_\xvec(\xvec;\thetavecsmall_m)}{f_\xvec(\xvec;\thetavecsmall)};\thetavec \right]=1$. 
The last equality is obtained by substituting the C-unbiasedness condition from \eqref{point_unbiased_cond_nec}. 
The advantage of the last term in \eqref{factorization} is that it is not a function of the estimator, as long as $\hat{\thetavec}$ is a C-unbiased estimator.
Substituting \eqref{zero_theta} and \eqref{factorization}  in \eqref{CS_step_1}, one obtains
\be \label{CS_step_2}
\begin{split}
{\text{WMSE}}_{\hat{\thetavecsmall}}(\thetavec)\cdot(\lambdavec^T\Dmat\lambdavec)\geq\left(\sum_{m=1}^{P}\lambdavec_m^T\tvec^{(m)}\right)^2\\=(\lambdavec^T\cdot [\zerovec^T, \tvec^T]^T)^2,
\end{split}
\ee
where the last equality is obtained by substituting \eqref{Su_cvec_define_W} and since $\tvec^{(m)}$ is the $m$th subvector of $\tvec$, $\forall m=1,\ldots,P$. The size of the zero vector in \eqref{CS_step_2} is
$M-K$.
 By using an extension of Cauchy-Schwarz inequality \cite[Eq. (2.37)]{PECARIC}, it can be verified that the tightest WMSE lower bound w.r.t. $\lambdavec$ that can be obtained from \eqref{CS_step_2} is for the choice
\be\label{lambda_opt_CS}
\lambdavec=\Dmat^{\dagger}\cdot [\zerovec^T, \tvec^T]^T.
\ee
Using \eqref{Su_cvec_define_W},
\eqref{lambda_opt_CS} can be written as
\beqna\label{lambda_opt_CS_m}
\lambdavec_0=\zerovec\hspace{4.95cm}\nonumber\\
\lambdavec_m=\left[\Dmat^{\dagger}\tvec\right]_{(m(M-K)+1):((m+1)(M-K))},
\eeqna
$\forall m=1,\ldots,P$. 
By substituting \eqref{lambda_opt_CS} in \eqref{CS_step_2} and using pseudo-inverse properties, we obtain
\be \label{total_CR_bound_CS_pre_W_app}
{\text{WMSE}}_{\hat{\thetavecsmall}}(\thetavec)\geq 
 [\zerovec^T, \tvec^T]\Dmat^{\dagger}[\zerovec^T, \tvec^T]^T.
\ee
It can be seen from \eqref{total_CR_bound_CS_pre_W_app} that only the lower-right $P(M-K)\times P(M-K)$ block of $\Dmat^{\dagger}$ affects the bound. Thus, we can write \eqref{total_CR_bound_CS_pre_W_app} as
the LU-CBTB from \eqref{total_CR_bound_CS_pre_W} when maximized w.r.t. the test-point vectors.\\

The Cauchy–Schwarz condition for equality in \eqref{CS_preliminary} with $\etavec$, $\xivec$,  and $\lambdavec_m$ from  \eqref{eta_term}, \eqref{xi_term}, and \eqref{lambda_opt_CS_m}, respectively, yields
\be \label{equality_cond_Prop_W_append}
\begin{split}
\Wmat^{\frac{1}{2}}(\hat{\thetavec}-\thetavec)&=\psi\Wmat^{\frac{1}{2}}\sum_{m=1}^{P}\bigg(\Umat(\thetavec_m)\frac{f_\xvec(\xvec;\thetavec_m)}{f_\xvec(\xvec;\thetavec)}\\
&~~~\times\bigg[[\Dmat^{\dagger}]_{LR} \tvec\bigg]_{((m-1)(M-K)+1):(m(M-K))}\bigg),
\end{split}
\ee
where 
$[ \Dmat^{\dagger}]_{LR}$ is the lower-right $P(M-K)\times P(M-K)$ block of $\Dmat^{\dagger}$.
for some scalar $\psi$ that may be parameter-dependent. It can be verified that in order for $\hat{\thetavec}$ from \eqref{equality_cond_Prop_W_append} to satisfy the C-unbiasedness conditions from \eqref{point_unbiased_cond_nec}, we must require 
\be\label{Su_zeta_solution}
\psi=1.
\ee
By inserting \eqref{Su_zeta_solution} into \eqref{equality_cond_Prop_W_append}, one obtains the equality condition from \eqref{equality_cond_Prop_W}.

\section{Proof of Proposition \ref{order_relation_prop}}\label{App_LU_CBTB_LU_CCRB_prop}
For the sake of simplicity we assume throughout this proof that the matrices $\Dmat(\thetavec,\Wmat,\Pimat)$, $\Umat^T(\thetavec)\Wmat\Umat(\thetavec)$, and $(\Umat^T(\thetavec)\Wmat\Umat(\thetavec))\otimes(\Umat^T(\thetavec)\Jmat(\thetavec)\Umat(\thetavec))+\Cmat_{\Umat,\Wmat}(\thetavec)$ are invertible. This assumption is for the sake of simplicity and is not a necessary condition for \eqref{order_relation} to hold.\\
\indent
In a similar manner to \cite{Hero_constraint}, we set $P=M-K$ and choose the test-point vectors
\be\label{tp_choice_proof}
\thetavec_m=\thetavec+\tau\uvec_m,
\ee
where $\tau\to0$ and $\uvec_m$ is the $m$th column of $\Umat(\thetavec)$, $\forall m=1,\ldots,M-K$. Thus, we choose the $M-K$ test-point vectors to be the true parameter vector with infinitesimal changes in the feasible directions \cite{BenHaim,sparse_con,normC,JOURNAL_CONSTRAINTS_EYAL} of the constrained set, $\Theta_\fvec$, at $\thetavec$. 

By substituting \eqref{tp_choice_proof} in \eqref{Dmat_c_block_define_boddy}, we obtain
that for this specific choice of test-point vectors, the matrix $[ \Dmat^{\dagger}(\thetavec,\Wmat,\Pimat)]_{LR}\in\mathbb{R}^{(M-K)^2\times (M-K)^2}$ from the LU-CBTB in \eqref{total_CR_bound_CS_pre_W}
can be obtained as the inverse of the
  block matrix 
whose $(m,n)$th block is given by
\be\label{Dmat_c_block_define}
\begin{split}
&\Dmat_c^{(m,n)}=\Umat^T(\thetavec+\tau\uvec_{m})\Wmat\Umat(\thetavec+\tau\uvec_{n})B_{m,n}\\
&~~~-\Umat^T(\thetavec+\tau\uvec_{m})\Wmat\Umat(\thetavec)(\Umat^T(\thetavec)\Wmat\Umat(\thetavec))^{-1}\\
&~~~\times\Umat^T(\thetavec)\Wmat\Umat(\thetavec+\tau\uvec_{n})
\\&=\Umat^T(\thetavec+\tau\uvec_{m})\Wmat\Umat(\thetavec+\tau\uvec_{n})(B_{m,n}-1)\\
&+\Umat^T(\thetavec+\tau\uvec_{m})\Big(\Wmat-\Wmat\Umat(\thetavec)[\Umat^T(\thetavec)\Wmat\Umat(\thetavec)]^{-1}\\
&\times\Umat^T(\thetavec)\Wmat\Big)\Umat(\thetavec+\tau\uvec_{n})\\
&=\Umat^T(\thetavec+\tau\uvec_{m})\Wmat\Umat(\thetavec+\tau\uvec_{n})(B_{m,n}-1)\\
&+\Big(\Umat(\thetavec+\tau\uvec_{m})-\Umat(\thetavec)\Big)^T\Big(\Wmat-\Wmat\Umat(\thetavec)\\
&\times[\Umat^T(\thetavec)\Wmat\Umat(\thetavec)]^{-1}\Umat^T(\thetavec)\Wmat\Big)\Big(\Umat(\thetavec+\tau\uvec_{n})+\Umat(\thetavec)\Big),
\end{split}
\ee
where the last equality is obtained since
\be\label{W_U_zero}
\Big(\Wmat-\Wmat\Umat(\thetavec)[\Umat^T(\thetavec)\Wmat\Umat(\thetavec)]^{-1}\Umat^T(\thetavec)\Wmat\Big)\Umat(\thetavec)=\zerovec.
\ee
\indent
First, we analyze the first term in the r.h.s. of \eqref{Dmat_c_block_define}. 
By using directional derivative properties \cite[p. 527]{CONVEX_BOYD}, it can be verified that
\be\label{C_directional_deriv}
\underset{\tau\to0}{\lim}\frac{1}{\tau}\Big(\frac{f_\xvec(\xvec;\thetavec+\tau\uvec_{m})}{f_\xvec(\xvec;\thetavec)}-1\Big)=\uvec_{m}^T\upsilonvec,~\forall m=1,\ldots,M-K,
\ee
where $\upsilonvec$ is defined in \eqref{l_define}. In addition, by using the linearity of the expectation operator on the right hand side (r.h.s.) of \eqref{Bmat_BTB_elements}, it can be shown that
\be\label{Cmat_BTB_elements}
B_{m,n}-1={\rm{E}}\bigg[\bigg(\frac{f_\xvec(\xvec;\thetavec_m)}{f_\xvec(\xvec;\thetavec)}-1\bigg)\bigg(\frac{f_\xvec(\xvec;\thetavec_n)}{f_\xvec(\xvec;\thetavec)}-1\bigg);\thetavec\bigg].
\ee
$\forall m,n=1,\ldots,M-K$. Thus, under Condition \ref{cond3CRB} and by substituting \eqref{C_directional_deriv} in \eqref{Cmat_BTB_elements}, we obtain
\be\label{Cmat_BTB_elements_lim}
\underset{\tau\to0}{\lim}\frac{1}{\tau^2}(B_{m,n}-1)=\uvec_{m}^T{\rm{E}}[\upsilonvec\upsilonvec^T;\thetavec]\uvec_{n}=\uvec_{m}^T\Jmat(\thetavec)\uvec_n,
\ee
where the second equality stems from \eqref{FIM}. Consequently, using \eqref{Cmat_BTB_elements_lim}, we obtain
\be\label{D_o_lim_m_n}
\begin{split}
&\underset{\tau\to0}{\lim}\frac{1}{\tau^2}\Umat^T(\thetavec+\tau\uvec_{m})\Wmat\Umat(\thetavec+\tau\uvec_{n})(B_{m,n}-1)\\
&=\uvec_{m}^T\Jmat(\thetavec)\uvec_n(\Umat^T(\thetavec)\Wmat\Umat(\thetavec)),
\end{split}
\ee
which is the $(m,n)$th block of the matrix $(\Umat^T(\thetavec)\Jmat(\thetavec)\Umat(\thetavec))\otimes(\Umat^T(\thetavec)\Wmat\Umat(\thetavec))$. Consider the vec-permutation matrix, $\Smat_{M-K}\in\mathbb{R}^{(M-K)^2\times (M-K)^2}$ \cite{HENDERSON_VEC}. This matrix satisfies 
\be\label{equal_vec}
\Smat_{M-K}=\Smat_{M-K}^T=\Smat_{M-K}^{-1},
\ee
\be\label{transpose_vec}
{\rm{vec}}(\Amat_1)=\Smat_{M-K}{\rm{vec}}(\Amat_1^T),
\ee
and
\be\label{kron_vec}
\Amat_1\otimes\Amat_2=\Smat_{M-K}(\Amat_2\otimes\Amat_1)\Smat_{M-K}^T,
\ee
for any $\Amat_1,\Amat_2\in\mathbb{R}^{(M-K)\times (M-K)}$ \cite{HENDERSON_VEC}. Using \eqref{kron_vec}, we can write
\be\label{D_o_lim}
\begin{split}
&(\Umat^T(\thetavec)\Jmat(\thetavec)\Umat(\thetavec))\otimes(\Umat^T(\thetavec)\Wmat\Umat(\thetavec))\\
&=\Smat_{M-K}\left((\Umat^T(\thetavec)\Wmat\Umat(\thetavec))\otimes(\Umat^T(\thetavec)\Jmat(\thetavec)\Umat(\thetavec))\right)\Smat_{M-K}^T.
\end{split}
\ee

Now, we analyze the second term in the r.h.s. of \eqref{Dmat_c_block_define}.  By using directional derivative properties \cite[p. 527]{CONVEX_BOYD}, it can be verified that
\be\label{U_directional_deriv}
\underset{\tau\to0}{\lim}\frac{1}{\tau}(\Umat(\thetavec+\tau\uvec_{m})-\Umat(\thetavec))=[\Vmat_1\uvec_{m},\ldots,\Vmat_{M-K}\uvec_{m}], 
\ee
where $\Vmat_m$ is defined in \eqref{V_define}, $\forall m=1,\ldots,M-K$. Using \eqref{U_directional_deriv}, we obtain
\be\label{Dmat_c_block_define_lim}
\begin{split}
&\underset{\tau\to0}{\lim}\frac{1}{\tau^2}\Big(\Umat(\thetavec+\tau\uvec_{m})-\Umat(\thetavec)\Big)^T\Big(\Wmat-\Wmat\Umat(\thetavec)\\
&\times(\Umat^T(\thetavec)\Wmat\Umat(\thetavec))^{-1}\Umat^T(\thetavec)\Wmat\Big)\Big(\Umat(\thetavec+\tau\uvec_{n})+\Umat(\thetavec)\Big)\\
&=[\Vmat_1\uvec_{m},\ldots,\Vmat_{M-K}\uvec_{m}]^T\\
&\times(\Wmat-\Wmat\Umat(\thetavec)(\Umat^T(\thetavec)\Wmat\Umat(\thetavec))^{-1}\Umat^T(\thetavec)\Wmat)\\
&\times[\Vmat_1\uvec_{n},\ldots,\Vmat_{M-K}\uvec_{n}]\\
& \define \Hmat^{(m,n)}.
\end{split}
\ee
Define the matrix $\Hmat\in\mathbb{R}^{(M-K)^2\times (M-K)^2}$ as a block matrix, whose $(m,n)$th block is given by $\Hmat^{(m,n)}$.
We can express $\Hmat$ in a full matrix form as follows
\be\label{Hmat_tao_matrix_define}
\begin{split}
&\Hmat=\begin{bmatrix}
[\Vmat_1\uvec_{1},\ldots,\Vmat_{M-K}\uvec_{1}]^T\\
\vdots\\
[\Vmat_1\uvec_{M-K},\ldots,\Vmat_{M-K}\uvec_{M-K}]^T
\end{bmatrix}\\
&\times(\Wmat-\Wmat\Umat(\thetavec)[\Umat^T(\thetavec)\Wmat\Umat(\thetavec)]^{-1}\Umat^T(\thetavec)\Wmat)\\
&\times[\Vmat_1\uvec_{1},..,\Vmat_{M-K}\uvec_{1},\ldots,\Vmat_1\uvec_{M-K},..,\Vmat_{M-K}\uvec_{M-K}].
\end{split}
\ee
Using \eqref{transpose_vec}, we obtain the following equality
\be\label{Hmat_Cmat_block_define}
\begin{split}
&\begin{bmatrix}[\Vmat_1\uvec_{1},\ldots,\Vmat_{M-K}\uvec_{1}]^T\\
\vdots\\
[\Vmat_1\uvec_{M-K},\ldots,\Vmat_{M-K}\uvec_{M-K}]^T
\end{bmatrix}\\
&=\Smat_{M-K}\begin{bmatrix}
\Umat^T(\thetavec)\Vmat_1^T\\
\vdots\\
\Umat^T(\thetavec)\Vmat_{M-K}^T
\end{bmatrix}.
\end{split}
\ee
By plugging \eqref{Hmat_Cmat_block_define} and \eqref{Gamma_block_define} into \eqref{Hmat_tao_matrix_define}, one obtains
\be\label{Hmat_tao_limit}
\Hmat=\Smat_{M-K}\Cmat_{\Umat,\Wmat}\Smat_{M-K}^T.
\ee
Substitution of \eqref{D_o_lim_m_n}, \eqref{D_o_lim}, \eqref{Dmat_c_block_define_lim}, and \eqref{Hmat_tao_limit} into the matrix form of \eqref{Dmat_c_block_define} yields
\be\label{Dmat_c_block_define_matrix}
\begin{split}
&\underset{\tau\to0}{\lim} \frac{1}{\tau^2}\Dmat_c\\
&=\Smat_{M-K}\left((\Umat^T(\thetavec)\Wmat\Umat(\thetavec))\otimes(\Umat^T(\thetavec)\Jmat(\thetavec)\Umat(\thetavec))\right)\Smat_{M-K}^T\\
&~~~+\Smat_{M-K}\Cmat_{\Umat,\Wmat}\Smat_{M-K}^T.
\end{split}
\ee

By inserting the specific choice of test-point vectors from \eqref{tp_choice_proof} into \eqref{tvec_block_define}, one obtains
that $\tvec\in\mathbb{R}^{(M-K)^2}$ is composed of $M-K$ subvectors:
\be\label{tvec_block_define_App_m}
\tvec^{(m)}=\tau\Umat^T(\thetavec+\tau\uvec_{m})\Wmat\uvec_m,~\forall ~m=1,\ldots,P.
\ee
Taking the limit $\tau\to0$ of $\frac{1}{\tau}\tvec$ based on the definition in \eqref{tvec_block_define_App_m}, results in
\beqna\label{t_c_tau_0_first}
\underset{\tau\to0}{\lim}\frac{1}{\tau}\tvec 
={\rm{vec}}(\Umat^T(\thetavec)\Wmat\Umat(\thetavec))\hspace{1.1cm}
\nonumber\\=\Smat_{M-K}{\rm{vec}}(\Umat^T(\thetavec)\Wmat\Umat(\thetavec)),
\eeqna
where the second equality stems from \eqref{transpose_vec}.

Finally, the LU-CBTB from \eqref{total_CR_bound_CS_pre_W} for the specific choice of test-point vectors from \eqref{tp_choice_proof}, denoted by $\tilde{B}_{\text{LU-CBTB}}$, is given by 
\be \label{total_CR_bound_CS_pre_W_CRB}
\tilde{B}_{\text{LU-CBTB}}=\frac{1}{\tau}\tvec^T\left(\frac{1}{\tau^2 }\Dmat_c\right)^{-1}\frac{1}{\tau}\tvec_c.
\ee
By inserting \eqref{t_c_tau_0_first} and \eqref{Dmat_c_block_define_matrix} into 
\eqref{total_CR_bound_CS_pre_W_CRB}, we get
\be \label{total_CR_bound_CS_pre_W_lim}
\begin{split}
&\underset{\tau\to0}{\lim}\tilde{B}_{\text{LU-CBTB}}={\rm{vec}}^T\left(\Umat^T(\thetavec)\Wmat\Umat(\thetavec)\right)\Smat_{M-K}^T\\
&\times\bigg(\Smat_{M-K}\left([\Umat^T(\thetavec)\Wmat\Umat(\thetavec)]\otimes[\Umat^T(\thetavec)\Jmat(\thetavec)\Umat(\thetavec)]\right)\Smat_{M-K}^T\\
&+\Smat_{M-K}\Cmat_{\Umat,\Wmat}\Smat_{M-K}^T\bigg)^{-1}\Smat_{M-K}{\rm{vec}}\left(\Umat^T(\thetavec)\Wmat\Umat(\thetavec)\right).
\end{split}
\ee
Using \eqref{equal_vec}, we can rewrite \eqref{total_CR_bound_CS_pre_W_lim} as
\be \label{total_CR_bound_CS_pre_W_res}
\begin{split}
&\underset{\tau\to0}{\lim}\tilde
{B}_{\text{LU-CBTB}}={\rm{vec}}^T\left(\Umat^T(\thetavec)\Wmat\Umat(\thetavec)\right)\\
&\times\bigg(\left([\Umat^T(\thetavec)\Wmat\Umat(\thetavec)]\otimes[\Umat^T(\thetavec)\Jmat(\thetavec)\Umat(\thetavec)]\right)+\Cmat_{\Umat,\Wmat}\bigg)^{-1}\\
&\times{\rm{vec}}\left(\Umat^T(\thetavec)\Wmat\Umat(\thetavec)\right)=B_{\text{LU-CCRB}},
\end{split}
\ee
where the second equality stems from \eqref{total_CR_bound_CS}.

The choice of test-point vectors, $\{\thetavec,\thetavec+\tau\uvec_1,\ldots,\thetavec+\tau\uvec_{M-K}\},~\tau\to0$, may not be the choice that maximizes the LU-CBTB and thus, \eqref{order_relation} is obtained.

\section{Proof of Proposition \ref{CBTB_order_prop}}\label{App_CBTB_order_prop}
Let $\tilde{B}_{\text{CBTB}}$ and $\tilde{B}_{\text{LU-CBTB}}$ denote the CBTB and the LU-CBTB from \eqref{CMS_W} and \eqref{total_CR_bound_CS_pre_W}, respectively, for a fixed choice of test-point vectors. Consider the estimator
\be \label{BTB_eff}
\hat\thetavec_{\text{CBTB}}=\thetavec+\Tmat\left(\Bmat-\onevec \onevec^T\right)^{\dagger}\begin{bmatrix}
\frac{f_\xvec(\xvec;\thetavecsmall_1)-f_\xvec(\xvec;\thetavecsmall)}{f_\xvec(\xvec;\thetavecsmall)}\\
\vdots\\
\frac{f_\xvec(\xvec;\thetavecsmall_P)-f_\xvec(\xvec;\thetavecsmall)}{f_\xvec(\xvec;\thetavecsmall)}
\end{bmatrix}.
\ee
It should be noted that this estimator is a function of $\thetavec$, i.e. it is defined separately for each fixed value of $\thetavec\in\Theta_\fvec$. Let $m\in\{1,\ldots,P\}$. By using \eqref{BTB_eff} and pseudo-inverse properties, we obtain
\be \label{point_unbiased}
{\rm{E}}[\hat{\thetavec}_{\text{CBTB}}-\thetavec_m;\thetavec_m]=\thetavec-\thetavec_m+\Tmat(\Bmat-\onevec \onevec^T)^{\dagger}(\bvec_{m}-\onevec).
\ee
where $\bvec_{m}$ is the $m$th column of $\Bmat$. Under Condition \ref{cond2CBTB} and by using pseudo-inverse properties, we obtain
\be\label{T_B_relation}
\Tmat((\Bmat -\onevec \onevec^T)^{\dagger}(\Bmat
-\onevec \onevec^T)-\Imat_{P})=\zerovec.
\ee
It can be observed that the vector on the r.h.s. of \eqref{point_unbiased} is the $m$th column of the matrix on the l.h.s. of \eqref{T_B_relation}, $\forall m=1,\ldots,P$. Thus, by substituting \eqref{T_B_relation} in \eqref{point_unbiased}, we get
\be \label{point_unbiased_pr}
{\rm{E}}[\hat{\thetavec}_{\text{CBTB}}-\thetavec_m;\thetavec_m]=\zerovec,~\forall m=1,\ldots,P.
\ee
It can be seen from \eqref{point_unbiased_pr} that the estimator from \eqref{BTB_eff} is pointwise mean-unbiased at the test-point vectors $\thetavec_1,\ldots,\thetavec_P\in\Theta_\fvec$, which implies that this estimator is also pointwise C-unbiased at $\thetavec_1,\ldots,\thetavec_P\in\Theta_\fvec$.
In particular,
by substituting \eqref{point_unbiased_pr} in the l.h.s. of \eqref{point_unbiased_cond_nec}, it can be observed that the estimator from \eqref{BTB_eff} satisfies \eqref{point_unbiased_cond_nec}.
In addition, it can be verified that the estimator from \eqref{BTB_eff} is point-wise mean-unbiased at $\thetavec$.
Therefore, for the specific choice of test-point vectors,
\be \label{pointwise_LU_CBTB}
{\text{WMSE}}_{\hat{\thetavecsmall}_{\text{CBTB}}}\geq \tilde{B}_{\text{LU-CBTB}}.
\ee

On the other hand, evaluating the WMSE, defined in \eqref{WMSE}, of $\hat\thetavec_{\text{CBTB}}$ from \eqref{BTB_eff} for the specific choice of test-point vectors yields
\beqna \label{pointwise_WMSE_CBTB}
{\text{WMSE}}_{\hat{\thetavecsmall}_{\text{CBTB}}}
={\rm{Tr}}\left({\rm{E}}[(\hat{\thetavec}_{\text{CBTB}}-\thetavec)(\hat{\thetavec}_{\text{CBTB}}-\thetavec)^T;\thetavec]\Wmat\right)
\nonumber\\
={\rm{Tr}}(\Tmat(\Bmat -\onevec \onevec^T)^{\dagger}\Tmat^T\Wmat)=\tilde{B}_{\text{CBTB}},\hspace{0.8cm}
\eeqna
where the second equality is obtained by substituting \eqref{Bmat_BTB_elements},  and using pseudo-inverse properties, and the last equality is obtained by substituting \eqref{CMS_W}.
By substituting \eqref{pointwise_WMSE_CBTB} in \eqref{pointwise_LU_CBTB}, one obtains that for any fixed choice of test-point vectors
\be\label{order_CBTBs_App}
\tilde{B}_{\text{LU-CBTB}}\leq \tilde{B}_{\text{CBTB}}
\ee
and consequently, \eqref{order_CBTBs} is obtained.

\section{Proof of Proposition \ref{linear_constraints_prop}}\label{App_linear_constraints_prop}
Let $\tilde{B}_{\text{CBTB}}$ and $\tilde{B}_{\text{LU-CBTB}}$ denote the CBTB and the LU-CBTB from \eqref{CMS_W} and \eqref{total_CR_bound_CS_pre_W}, respectively, for a fixed choice of test-point vectors. Under the linear constraint in \eqref{linear_constr_eq}, the unknown parameter vector belongs to the null space of $\Amat$, i.e. $\thetavec\in{\mathcal{N}}(\Amat)$. In this case $\Fmat(\thetavec)=\Amat$, which is not a function of $\thetavec$. Therefore, the orthonormal complement matrix $\Umat$ is not a function of $\thetavec$ as well. According to \eqref{set_def}-\eqref{two}, it can be seen that 
\be\label{linear_param_space}
\Theta_\fvec={\mathcal{R}}(\Umat) 
\ee
and thus,
\be\label{U_linear_equality}
\thetavec=\Pmat_\Umat\thetavec+\Pmat_\Umat^\bot\thetavec=\Pmat_\Umat\thetavec=\Umat\Umat^T\thetavec,~\forall\thetavec\in\Theta_\fvec,
\ee
where the second equality stems from \eqref{linear_param_space} and the last equality stems from the definition of orthogonal projection matrix. From \eqref{U_linear_equality}, we obtain
\be\label{U_linear_equality_T}
\Umat\Umat^T(\thetavec_m-\thetavec)=\thetavec_m-\thetavec,
\ee
where $\thetavec_m-\thetavec$ is the $m$th column of $\Tvec$ from \eqref{T_BTB}, $\forall m=1,\ldots,P$. Consequently, we get
\be\label{U_T_0_linear_equality}
\Umat\Umat^T\Tmat=\Tmat.
\ee
Since the matrix $\Umat$ is not a function of $\thetavec$, by using \eqref{Dmat_block_define}, Kronecker product rules, and pseudo-inverse properties, we can express the matrix $\Dmat^{\dagger}$ as follows
\be\label{D_W_linear}
\Dmat^{\dagger}=\left[\begin{array}{cc} 1 &\onevec^T\\
\onevec & \Bmat
\end{array}\right]^{\dagger}\otimes (\Umat^T\Wmat\Umat)^{\dagger}.
\ee
By substituting \eqref{U_linear_equality_T} in \eqref{tvec_block_define}, one obtains
\be\label{tvec_block_define_AppD}
\tvec^{(m)}=(\Umat^T\Wmat\Umat)\Umat^T(\thetavec_m-\thetavec),
\ee
which is the the $m$th column of the matrix $(\Umat^T\Wmat\Umat)\Umat^T\Tmat,~\forall m=1,\ldots,P$. Consequently, we can express the vector $\tvec$ as follows
\be\label{t_W_linear}
\tvec={\rm{vec}}((\Umat^T\Wmat\Umat)\Umat^T\Tmat).
\ee
By inserting \eqref{D_W_linear} and \eqref{t_W_linear} into $\tilde{B}_{\text{LU-CBTB}}$ from \eqref{total_CR_bound_CS_pre_W} (without the maximization w.r.t. the test-point vectors) under the linear constraints, we obtain
\be \label{total_CR_bound_CS_pre_W_linear}
\begin{split}
\tilde{B}_{\text{LU-CBTB}}&={\rm{vec}}^T((\Umat^T\Wmat\Umat)\Umat^T\Tmat)((\Bmat-\onevec \onevec^T)^{\dagger}\otimes (\Umat^T\Wmat\Umat)^{\dagger})\\
&~~~\times{\rm{vec}}((\Umat^T\Wmat\Umat)\Umat^T\Tmat).
\end{split}
\ee
By substituting the equality \cite[p. 60]{MATRIX_COOKBOOK}
\be\label{kronecker_prop}
{\rm{vec}}^T(\Amat_4)(\Amat_2^T\otimes\Amat_1){\rm{vec}}(\Amat_3)={\rm{Tr}}(\Amat_4^T\Amat_1\Amat_3\Amat_2)
\ee
with $\Amat_1=(\Umat^T\Wmat\Umat)^{\dagger}$, $\Amat_2=(\Bmat-\onevec \onevec^T)^{\dagger}$, and $\Amat_3=\Amat_4=(\Umat^T\Wmat\Umat)\Umat^T\Tmat$, in \eqref{total_CR_bound_CS_pre_W_linear}, one obtains
\be \label{total_CR_bound_CS_pre_W_linear_substitute}
\begin{split}
&\tilde{B}_{\text{LU-CBTB}}\\
&={\rm{Tr}}\left(\Tmat^T\Umat(\Umat^T\Wmat\Umat)
\right.\\ & \times \left.(\Umat^T\Wmat\Umat)^{\dagger}(\Umat^T\Wmat\Umat)\Umat^T\Tmat(\Bmat-\onevec \onevec^T)^{\dagger}\right)\\
&={\rm{Tr}}\left(\Tmat(\Bmat-\onevec \onevec^T)^{\dagger}\Tmat^T\Wmat\right)=\tilde{B}_{\text{CBTB}},
\end{split}
\ee
where the second equality stems from \eqref{U_T_0_linear_equality} and using pseudo-inverse and trace properties. The last equality stems from \eqref{CMS_W} for the fixed choice of test-point vectors. Thus, from \eqref{total_CR_bound_CS_pre_W_linear_substitute} we get \eqref{linear_constraint_equality}.

\section{Proof of Claim \ref{claim_sim}}
\label{new_app}
By substituting 
the test-point vectors from \eqref{tp_nu}-\eqref{tp_alpha} into \eqref{Bmat_BTB_elements} and using Eq. (16) from \cite{TABRIKIAN_OCEAN}, we obtain that the matrix $\Bmat$ is given by:
\beqna
\Bmat\define
\begin{bmatrix}
\bar{B}_{2,2} & \bar{B}_{2,3} & \bar{B}_{2,4}\\
 \bar{B}_{2,3}& \bar{B}_{3,3} & \bar{B}_{3,4}\\
 \bar{B}_{2,4}& \bar{B}_{3,4}& \bar{B}_{4,4}
\end{bmatrix}+\onevec \onevec^T,
\eeqna
where 
\begin{subequations}
\be\label{Bmat_22}
\bar{B}_{2,2}=\gamma (-2\alpha,-2\alpha,d_q,0),
\ee
\be\label{Bmat_23}
\begin{split}
B_{2,3}=\gamma (-\alpha,-2\alpha,h_\phi,-\cos{d_q}+\cos(h_\phi-d_q)),
\end{split}
\ee
\be\label{Bmat_24}
\bar{B}_{2,4}=
\gamma (h_\alpha,-2\alpha,d_q,0),
\ee
\be\label{Bmat_33}
\bar{B}_{3,3}=\gamma (-2\alpha  ,-2\alpha,h_\phi,0),
\ee
\be\label{Bmat_34}
\bar{B}_{3,4}=\gamma (h_\alpha,-2 \alpha ,h_\phi,0),
\ee
\be\label{Bmat_44}
\bar{B}_{4,4}=\gamma (h_{\alpha}, h_{\alpha},-\pi,0) ,
\ee
\end{subequations}
in which $d_q\define
\angle a_{q}(\nu \cdot e^{jh_\nu})
-\angle a_{q}(\nu)$,
$\forall q=1,\ldots,Q$,
 $a_{q}(\nu)$ is defined in \eqref{a_q_def},
and
\begin{equation}
\gamma (\kappa_1,\kappa_2,\kappa_3,\kappa_4) 
\define 
\exp\left(\frac{ \kappa_1 \kappa_2}{\sigma^2}\sum_{q=1}^{Q}(1-\cos{\kappa_3}+\kappa_4)\right)-1. 
\end{equation}
In general, $\kappa_3$ and $\kappa_4$ are functions of $q$.
\\
\indent
Next, we derive the CBTB from \eqref{CMS_W}. By inserting the test-point vectors from \eqref{tp_nu}-\eqref{tp_alpha} in \eqref{T_BTB} and reordering, one obtains
\be\label{Tmat_DOA}
\begin{split}
\Tmat=\begin{bmatrix}
 \cos(\angle{\nu}+h_\nu)-\cos\angle{\nu} & 0 & 0\\
 \sin(\angle{\nu}+h_\nu)-\sin\angle{\nu} & 0 & 0\\
 0 & h_\phi & 0\\
 0 & 0 & h_\alpha
\end{bmatrix}.
\end{split}
\ee
By substituting $\Wmat$ from \eqref{W_choice}, the elements of $\Bmat$ from \eqref{Bmat_22}-\eqref{Bmat_44}, and \eqref{Tmat_DOA} into \eqref{CMS_W} and applying matrix inversion rules \cite{MATRIX_COOKBOOK}, we obtain a closed-form expression of the CBTB, given in \eqref{CBTB_DOA}.\\
\indent
In order to derive the LU-CBTB from \eqref{total_CR_bound_CS_pre_W}, we substitute $\Wmat$ from \eqref{W_choice}, $\thetavec_0=\thetavec$, the test-point vectors from \eqref{tp_nu}-\eqref{tp_alpha}, and the elements of $\Bmat$ into \eqref{Dmat_block_define} and use the Kronecker product definition to obtain
\beqna\label{Dmat_DOA}
\Dmat=\hspace{6cm}
\nonumber\\
\left(\left[\begin{array}{cc} 1 &\onevec^T\\
\onevec & \Bmat
\end{array}\right] \odot
\begin{bmatrix}
1 & \cos h_\nu & 1 & 1\\
\cos h_\nu & 1 & \cos h_\nu & \cos h_\nu\\
1 & \cos h_\nu &1 & 1\\
1 & \cos h_\nu & 1 & 1
\end{bmatrix}\right)
\nonumber\\
\otimes(\evec_1^{(3)}(\evec_1^{(3)})^T),
\eeqna
where $\evec_1^{(3)}=[1,0,0]^T$ is the first column of $\Imat_3$.
Then, by plugging $\Wmat$ from \eqref{W_choice} and the test-point vectors from \eqref{tp_nu}-\eqref{tp_alpha} in \eqref{tvec_block_define} and reordering, one obtains
\be\label{tvec_DOA}
\tvec={\rm{vec}}(\evec_1^{(3)}[-\sin{h_\nu},0,0]).
\ee
Substituting \eqref{Dmat_DOA} and \eqref{tvec_DOA} into \eqref{total_CR_bound_CS_pre_W} and applying matrix inversion and Kronecker product rules \cite{MATRIX_COOKBOOK}, we obtain a closed-form expression of the LU-CBTB in \eqref{LU_CBTB_DOA}.


\end{document}

%% file: myshorts.tex
\newcommand{\avec}{{\bf{a}}}
\newcommand{\bvec}{{\bf{b}}}

\newcommand{\evec}{{\bf{e}}}
\newcommand{\fvec}{{\bf{f}}}

\newcommand{\uvec}{{\bf{u}}}

\newcommand{\xvec}{{\bf{x}}}

\newcommand{\nvec}{{\bf{n}}}
\newcommand{\tvec}{{\bf{t}}}

\newcommand{\Tvec}{{\bf{T}}}

\newcommand{\gvec}{{\bf{g}}}

\newcommand{\etavec}{{\bf{\eta}}}

\newcommand{\onevec}{{\bf{1}}}
\newcommand{\zerovec}{{\bf{0}}}

\newcommand{\Amat}{{\bf{A}}}
\newcommand{\Bmat}{{\bf{B}}}
\newcommand{\Cmat}{{\bf{C}}}
\newcommand{\Dmat}{{\bf{D}}}

\newcommand{\Fmat}{{\bf{F}}}

\newcommand{\Hmat}{{\bf{H}}}
\newcommand{\Jmat}{{\bf{J}}}

\newcommand{\Imat}{{\bf{I}}}

\newcommand{\Pmat}{{\bf{P}}}

\newcommand{\Smat}{{\bf{S}}}
\newcommand{\Tmat}{{\bf{T}}}

\newcommand{\Umat}{{\bf{U}}}
\newcommand{\Vmat}{{\bf{V}}}
\newcommand{\Wmat}{{\bf{W}}}

\newcommand{\define}{\stackrel{\triangle}{=}}

\newcommand{\Pimat}{\mbox{\boldmath $\Pi$}}

\def\Pimatsmall{{\mbox{\boldmath {\scriptsize $\Pi$}}}}

\def\upsilonvec{{\mbox{\boldmath $\upsilon$}}}

\def\xivec{{\mbox{\boldmath $\xi$}}}

\def\etavec{{\mbox{\boldmath $\eta$}}}

\def\thetavec{{\mbox{\boldmath $\theta$}}}

\def\lambdavec{{\mbox{\boldmath $\lambda$}}}

\def\thetavecsmall{{\mbox{\boldmath {\scriptsize $\theta$}}}}

\newcommand{\be}{\begin{equation}}
\newcommand{\ee}{\end{equation}}
\newcommand{\beqna}{\begin{eqnarray}}
\newcommand{\eeqna}{\end{eqnarray}}


%% file: main_arxiv.bbl
\begin{thebibliography}{10}
\providecommand{\url}[1]{#1}
\csname url@samestyle\endcsname
\providecommand{\newblock}{\relax}
\providecommand{\bibinfo}[2]{#2}
\providecommand{\BIBentrySTDinterwordspacing}{\spaceskip=0pt\relax}
\providecommand{\BIBentryALTinterwordstretchfactor}{4}
\providecommand{\BIBentryALTinterwordspacing}{\spaceskip=\fontdimen2\font plus
\BIBentryALTinterwordstretchfactor\fontdimen3\font minus
  \fontdimen4\font\relax}
\providecommand{\BIBforeignlanguage}[2]{{%
\expandafter\ifx\csname l@#1\endcsname\relax
\typeout{** WARNING: IEEEtran.bst: No hyphenation pattern has been}%
\typeout{** loaded for the language `#1'. Using the pattern for}%
\typeout{** the default language instead.}%
\else
\language=\csname l@#1\endcsname
\fi
#2}}
\providecommand{\BIBdecl}{\relax}
\BIBdecl

\bibitem{Hero_constraint}
J.~D. Gorman and A.~O. Hero, ``Lower bounds for parametric estimation with
  constraints,'' \emph{IEEE Trans. Inf. Theory}, vol.~36, no.~6, pp.
  1285--1301, Nov. 1990.

\bibitem{Stoica_Ng}
P.~Stoica and B.~C. Ng, ``On the {C}ram$\acute{\text{e}}$r-{R}ao bound under
  parametric constraints,'' \emph{IEEE Signal Process. Lett.}, vol.~5, no.~7,
  pp. 177--179, July 1998.

\bibitem{HERO_EMISSION}
A.~O. Hero, R.~Piramuthu, J.~A. Fessler, and S.~R. Titus, ``Minimax emission
  computed tomography using high-resolution anatomical side information and
  {B}-spline models,'' \emph{IEEE Trans. Inf. Theory}, vol.~45, no.~3, pp.
  920--938, Apr. 1999.

\bibitem{SADLER_sim}
B.~M. Sadler, R.~J. Kozick, and T.~Moore, ``Bounds on bearing and symbol
  estimation with side information,'' \emph{IEEE Trans. Signal Process.},
  vol.~49, no.~4, pp. 822--834, Apr. 2001.

\bibitem{WIJNHOLDS_sim}
S.~J. Wijnholds and A.~J. van~der Veen, ``Effects of parametric constraints on
  the {CRLB} in gain and phase estimation problems,'' \emph{IEEE Signal
  Process. Lett.}, vol.~13, no.~10, pp. 620--623, Oct. 2006.

\bibitem{ROUTTENBERG_sim}
T.~Routtenberg and L.~Tong, ``Joint frequency and phasor estimation under the
  {KCL} constraint,'' \emph{IEEE Signal Process. Lett.}, vol.~20, no.~6, pp.
  575--578, June 2013.

\bibitem{HERO_HYPER}
N.~Dobigeon, J.~Y. Tourneret, C.~Richard, J.~C.~M. Bermudez, S.~McLaughlin, and
  A.~O. Hero, ``Nonlinear unmixing of hyperspectral images: Models and
  algorithms,'' \emph{IEEE Signal Process. Mag.}, vol.~31, no.~1, pp. 82--94,
  Jan. 2014.

\bibitem{MENNI_sim}
T.~Menni, J.~Galy, E.~Chaumette, and P.~Larzabal, ``Versatility of constrained
  {CRB} for system analysis and design,'' \emph{IEEE Trans. Aerosp. Electron.
  Syst.}, vol.~50, no.~3, pp. 1841--1863, July 2014.

\bibitem{SHIR_PARAMETRIC}
S.~Cohen, T.~Routtenberg, and L.~Tong, ``Non-bayesian parametric missing-mass
  estimation,'' \emph{IEEE Trans. Signal Processing}, vol.~70, pp. 3709--3725,
  2022.

\bibitem{PREVOST_CONSTRAINED_CRAMER}
C.~Pr{\'e}vost, K.~Usevich, M.~Haardt, P.~Comon, and D.~Brie, ``Constrained
  {C}ram{\'e}r--{R}ao bounds for reconstruction problems formulated as coupled
  canonical polyadic decompositions,'' \emph{Signal Processing}, vol. 198,
  2022.

\bibitem{Marzetta}
T.~L. Marzetta, ``A simple derivation of the constrained multiple parameter
  {C}ram$\acute{\text{e}}$r-{R}ao bound,'' \emph{IEEE Trans. Signal Process.},
  vol.~41, no.~6, pp. 2247--2249, June 1993.

\bibitem{CCRB_complex}
A.~K. Jagannatham and B.~D. Rao, ``Cram$\acute{\text{e}}$r-{R}ao lower bound
  for constrained complex parameters,'' \emph{IEEE Signal Process. Lett.},
  vol.~11, no.~11, pp. 875--878, Nov. 2004.

\bibitem{BenHaim}
Z.~Ben-Haim and Y.~C. Eldar, ``On the constrained
  {C}ram$\acute{\text{e}}$r-{R}ao bound with a singular {F}isher information
  matrix,'' \emph{IEEE Signal Process. Lett.}, vol.~16, no.~6, pp. 453--456,
  June 2009.

\bibitem{sparse_con}
------, ``The {C}ram$\acute{\text{e}}$r-{R}ao bound for estimating a sparse
  parameter vector,'' \emph{IEEE Trans. Signal Process.}, vol.~58, no.~6, pp.
  3384--3389, June 2010.

\bibitem{Moore_fitting}
T.~Moore, R.~Kozick, and B.~Sadler, ``The constrained
  {C}ram$\acute{\text{e}}$r-{R}ao bound from the perspective of fitting a
  model,'' \emph{IEEE Signal Process. Lett.}, vol.~14, no.~8, pp. 564--567,
  Aug. 2007.

\bibitem{Moore_scoring}
T.~J. Moore, B.~M. Sadler, and R.~J. Kozick, ``Maximum-likelihood estimation,
  the {C}ram$\acute{\text{e}}$r-{R}ao bound, and the method of scoring with
  parameter constraints,'' \emph{IEEE Trans. Signal Process.}, vol.~56, no.~3,
  pp. 895--908, Mar. 2008.

\bibitem{normC}
S.~Zhiguang, Z.~Jianxiong, H.~Lei, and L.~Jicheng, ``A new derivation of
  constrained {C}ram$\acute{\text{e}}$r-{R}ao bound via norm minimization,''
  \emph{IEEE Trans. Signal Process.}, vol.~59, no.~4, pp. 1879--1882, Apr.
  2011.

\bibitem{SOMEKH_LESHEM}
A.~Somekh-Baruch, A.~Leshem, and V.~Saligrama, ``On the non-existence of
  unbiased estimators in constrained estimation problems,'' \emph{IEEE Trans.
  Inf. Theory}, vol.~64, no.~8, pp. 5549--5554, Aug. 2018.

\bibitem{JOURNAL_CONSTRAINTS_EYAL}
E.~Nitzan, T.~Routtenberg, and J.~Tabrikian, ``Cram$\acute{\text{e}}$r-{R}ao
  bound for constrained parameter estimation using {L}ehmann-unbiasedness,''
  \emph{IEEE Trans. Signal Process.}, vol.~67, no.~3, pp. 753--768, Feb. 2019.

\bibitem{SSP_EYAL}
------, ``Limitations of constrained {CRB} and an alternative bound,'' in
  \emph{Proc. of the IEEE Statistical Signal Processing Workshop (SSP)}, June
  2018, pp. 673--677.

\bibitem{point_est}
E.~L. Lehmann and G.~Casella, \emph{Theory of Point Estimation (Springer Texts
  in Statistics)}, 2nd~ed.\hskip 1em plus 0.5em minus 0.4em\relax Springer,
  1998.

\bibitem{SAM2012constraint}
T.~Routtenberg and J.~Tabrikian, ``Performance bounds for constrained parameter
  estimation,'' in \emph{Proc. of the 7th IEEE Sensor Array and Multichannel
  Signal Processing Workshop (SAM)}, June 2012, pp. 513--516.

\bibitem{PCRB}
------, ``Non-{B}ayesian periodic {C}ram$\acute{\text{e}}$r-{R}ao bound,''
  \emph{IEEE Trans. Signal Process.}, vol.~61, no.~4, pp. 1019--1032, Feb.
  2013.

\bibitem{CYCLIC}
------, ``Cyclic {B}arankin-type bounds for non-{B}ayesian periodic parameter
  estimation,'' \emph{IEEE Trans. Signal Process.}, vol.~62, no.~13, pp.
  3321--3336, July 2014.

\bibitem{BAR1}
S.~Bar and J.~Tabrikian, ``Bayesian estimation in the presence of deterministic
  nuisance parameters--part {I}: Performance bounds,'' \emph{IEEE Trans. Signal
  Process.}, vol.~63, no.~24, pp. 6632--6646, Dec. 2015.

\bibitem{SELECTION}
T.~Routtenberg and L.~Tong, ``Estimation after parameter selection: Performance
  analysis and estimation methods,'' \emph{IEEE Trans. Signal Process.},
  vol.~64, no.~20, pp. 5268--5281, Oct. 2016.

\bibitem{BAR_NB}
S.~Bar and J.~Tabrikian, ``The risk-unbiased {C}ram$\acute{\text{e}}$r-{R}ao
  bound for non-{B}ayesian multivariate parameter estimation,'' \emph{IEEE
  Trans. Signal Process.}, vol.~66, no.~18, pp. 4920--4934, Sep. 2018.

\bibitem{LETTER_NORM_EYAL}
E.~{Nitzan}, T.~{Routtenberg}, and J.~{Tabrikian},
  ``Cram$\acute{\text{e}}$r-{R}ao bound under norm constraint,'' \emph{IEEE
  Signal Process. Lett.}, vol.~26, no.~9, pp. 1393--1397, Sep. 2019.

\bibitem{MEIR_ROUTTENBERG}
E.~Meir and T.~Routtenberg, ``Cram$\acute{\text{e}}$r-{R}ao bound for
  estimation after model selection and its application to sparse vector
  estimation,'' \emph{IEEE Trans. Signal Process.}, vol.~69, pp. 2284--2301,
  2021.

\bibitem{HAMMERSLEY}
J.~M. Hammersley, ``On estimating restricted parameters,'' \emph{Journal of the
  Royal Statistical Society. Series B (Methodological)}, vol.~12, no.~2, pp.
  192--240, 1950.

\bibitem{HCR}
D.~G. Chapman and H.~Robbins, ``Minimum variance estimation without regularity
  assumptions,'' \emph{The Annals of Mathematical Statistics}, pp. 581--586,
  1951.

\bibitem{TTB1}
K.~Todros and J.~Tabrikian, ``General classes of performance lower bounds for
  parameter estimation part {I}: Non-{B}ayesian bounds for unbiased
  estimators,'' \emph{IEEE Trans. Inf. Theory}, vol.~56, no.~10, pp.
  5045--5063, Oct. 2010.

\bibitem{HASSIBI_BOYD}
A.~{Hassibi} and S.~{Boyd}, ``Integer parameter estimation in linear models
  with applications to {GPS},'' \emph{IEEE Trans. Signal Process.}, vol.~46,
  no.~11, pp. 2938--2952, Nov. 1998.

\bibitem{LAROSA_CHANGE_POINT}
P.~S. {La Rosa}, A.~{Renaux}, C.~H. {Muravchik}, and A.~{Nehorai},
  ``Barankin-type lower bound on multiple change-point estimation,'' \emph{IEEE
  Trans. Signal Process.}, vol.~58, no.~11, pp. 5534--5549, Nov. 2010.

\bibitem{BARANKIN}
E.~W. Barankin, ``Locally best unbiased estimates,'' \emph{Ann. Math. Stat.},
  vol.~20, pp. 477--501, 1946.

\bibitem{MS}
R.~J. McAulay and L.~P. Seidman, ``A useful form of the {B}arankin lower bound
  and its application to {P}{P}{M} threshold analysis,'' \emph{IEEE Trans. Inf.
  Theory}, vol.~15, pp. 273--279, 1969.

\bibitem{REUVEN_MESSER}
I.~{Reuven} and H.~{Messer}, ``A {B}arankin-type lower bound on the estimation
  error of a hybrid parameter vector,'' \emph{IEEE Trans. Inf. Theory},
  vol.~43, no.~3, pp. 1084--1093, May 1997.

\bibitem{JUNG_SPARSE}
A.~{Jung}, Z.~{Ben-Haim}, F.~{Hlawatsch}, and Y.~C. {Eldar}, ``Unbiased
  estimation of a sparse vector in white {G}aussian noise,'' \emph{IEEE Trans.
  Inf. Theory}, vol.~57, no.~12, pp. 7856--7876, Dec. 2011.

\bibitem{ziv1969some}
J.~Ziv and M.~Zakai, ``Some lower bounds on signal parameter estimation,''
  \emph{IEEE Trans. Inf. Theory}, vol.~15, no.~3, pp. 386--391, 1969.

\bibitem{Aharon_Tabrikian}
O.~Aharon and J.~Tabrikian, ``A class of {B}ayesian lower bounds for parameter
  estimation via arbitrary test-point transformation,'' \emph{IEEE Trans.
  Signal Processing}, vol.~71, pp. 2296--2308, 2023.

\bibitem{EXTENDED_ZZ}
K.~L. Bell, Y.~Steinberg, Y.~Ephraim, and H.~L. Van~Trees, ``Extended
  {Z}iv-{Z}akai lower bound for vector parameter estimation,'' \emph{IEEE
  Trans. Inf. Theory}, vol.~43, no.~2, pp. 624--637, Mar. 1997.

\bibitem{CAMPBELL}
S.~L. Campbell and C.~D. Meyer, \emph{Generalized inverses of linear
  transformations}.\hskip 1em plus 0.5em minus 0.4em\relax SIAM, 2009.

\bibitem{PILZ}
J.~Pilz, ``Minimax linear regression estimation with symmetric parameter
  restrictions,'' \emph{Journal of Statistical Planning and Inference},
  vol.~13, pp. 297--318, 1986.

\bibitem{ELDAR_WEIGHTED_MSE}
Y.~C. Eldar, ``Universal weighted {MSE} improvement of the least-squares
  estimator,'' \emph{IEEE Trans. Signal Process.}, vol.~56, no.~5, pp.
  1788--1800, May 2008.

\bibitem{LEHMANN_CONCEPT}
E.~L. Lehmann, ``A general concept of unbiasedness,'' \emph{The Annals of
  Mathematical Statistics}, vol.~22, no.~4, pp. 587--592, Dec. 1951.

\bibitem{BETTS}
J.~T. Betts, \emph{Practical Methods for Optimal Control Using Nonlinear
  Programming}.\hskip 1em plus 0.5em minus 0.4em\relax SIAM, 2001.

\bibitem{MATRIX_COOKBOOK}
K.~B. Petersen and M.~S. Pedersen, ``The matrix cookbook. version: November 15,
  2012,'' 2012.

\bibitem{VAN_DER_VEEN_CM}
A.~J. van~der Veen and A.~Paulraj, ``An analytical constant modulus
  algorithm,'' \emph{IEEE Trans. Signal Process.}, vol.~44, no.~5, pp.
  1136--1155, May 1996.

\bibitem{LESHEM_CM}
A.~Leshem and A.~J. van~der Veen, ``Direction-of-arrival estimation for
  constant modulus signals,'' \emph{IEEE Trans. Signal Process.}, vol.~47,
  no.~11, pp. 3125--3129, Nov. 1999.

\bibitem{STOICA_CM}
P.~Stoica and O.~Besson, ``Maximum likelihood {DOA} estimation for
  constant-modulus signal,'' \emph{Electronics Letters}, vol.~36, no.~9, pp.
  849--851, Apr. 2000.

\bibitem{GAMBOA}
F.~{Gamboa} and E.~{Gassiat}, ``Source separation when the input sources are
  discrete or have constant modulus,'' \emph{IEEE Trans. Signal Process.},
  vol.~45, no.~12, pp. 3062--3072, 1997.

\bibitem{DELMAS_DOA}
J.~P. {Delmas} and H.~{Abeida}, ``Cram$\acute{\text{e}}$r-{R}ao bounds of {DOA}
  estimates for {BPSK} and {QPSK} modulated signals,'' \emph{IEEE Trans. Signal
  Process.}, vol.~54, no.~1, pp. 117--126, 2006.

\bibitem{DELMAS_DOA_SIGNALS}
J.~P. {Delmas}, ``Closed-form expressions of the exact
  {C}ram$\acute{\text{e}}$r-{R}ao bound for parameter estimation of {BPSK},
  {MSK}, or {QPSK} waveforms,'' \emph{IEEE Signal Process. Lett.}, vol.~15, pp.
  405--408, 2008.

\bibitem{KRUMMENAUER}
R.~Krummenauer, R.~Ferrari, R.~Suyama, R.~Attux, C.~Junqueira, P.~Larzabal,
  P.~Forster, and A.~Lopes, ``Maximum likelihood-based direction-of-arrival
  estimator for discrete sources,'' \emph{Circuits, Systems, and Signal
  Processing}, vol.~32, no.~5, pp. 2423--2443, 2013.

\bibitem{MEDINA_INTEGER}
D.~Medina, J.~Vil{\`a}-Valls, E.~Chaumette, F.~Vincent, and P.~Closas,
  ``Cram{\'e}r-{R}ao bound for a mixture of real-and integer-valued parameter
  vectors and its application to the linear regression model,'' \emph{Signal
  Processing}, 2020.

\bibitem{LESHEM_DISCRETE}
A.~{Leshem} and A.~J. {van der Veen}, ``On the number of samples needed to
  identify a mixture of finite alphabet constant modulus sources,'' in
  \emph{Proc. of the IEEE International Conference on Acoustics, Speech, and
  Signal Processing (ICASSP)}, vol.~4, 2003, pp. IV--329.

\bibitem{MANIOUDAKIS}
S.~{Manioudakis}, ``Blind estimation of space-time coded systems using the
  finite alphabet-constant modulus algorithm,'' \emph{IEE Proceedings - Vision,
  Image and Signal Processing}, vol. 153, no.~5, pp. 549--556, 2006.

\bibitem{YANG_DISCRETE}
J.~{Yang}, A.~{Aubry}, A.~{De Maio}, X.~{Yu}, and G.~{Cui}, ``Design of
  constant modulus discrete phase radar waveforms subject to multi-spectral
  constraints,'' \emph{IEEE Signal Process. Lett.}, vol.~27, pp. 875--879,
  2020.

\bibitem{Golub}
G.~H. Golub and U.~Von~Matt, ``Quadratically constrained least squares and
  quadratic problems,'' \emph{Numerische Mathematik}, vol.~59, no.~1, pp.
  561--580, 1991.

\bibitem{CHEN_REGULAR}
S.~Chen, E.~S. Chng, and K.~Alkadhimi, ``Regularized orthogonal least squares
  algorithm for constructing radial basis function networks,''
  \emph{International Journal of Control}, vol.~64, no.~5, pp. 829--837, 1996.

\bibitem{PHASE_KAY}
K.~Peters and S.~Kay, ``Unbiased estimation of the phase of a sinusoid,'' in
  \emph{Proc. of the IEEE International Conference on Acoustics, Speech, and
  Signal Processing (ICASSP)}, vol.~2, May 2004, pp. 493--496.

\bibitem{TODROS_WINNIK}
K.~Todros, R.~Winik, and J.~Tabrikian, ``On the limitations of {B}arankin type
  bounds for {MLE} threshold prediction,'' \emph{Signal Processing}, vol. 108,
  pp. 622--627, 2015.

\bibitem{DIRECTIONAL}
K.~V. Mardia and P.~E. Jupp, \emph{Directional Statistics}, ser. Wiley Series
  in Probability and Statistics.\hskip 1em plus 0.5em minus 0.4em\relax
  Chichester: Wiley, 1999.

\bibitem{RIFE}
D.~Rife and R.~Boorstyn, ``Single tone parameter estimation from discrete-time
  observations,'' \emph{IEEE Trans. Inf. Theory}, vol.~20, no.~5, pp. 591--598,
  Sep. 1974.

\bibitem{Richmond_2005}
C.~Richmond, ``Capon algorithm mean-squared error threshold {SNR} prediction
  and probability of resolution,'' \emph{IEEE Trans. Signal Processing},
  vol.~53, no.~8, pp. 2748--2764, 2005.

\bibitem{TABRIKIAN_OCEAN}
J.~{Tabrikian} and J.~L. {Krolik}, ``Barankin bounds for source localization in
  an uncertain ocean environment,'' \emph{IEEE Trans. Signal Process.},
  vol.~47, no.~11, pp. 2917--2927, Nov. 1999.

\bibitem{BZB_EYAL}
E.~{Nitzan}, T.~{Routtenberg}, and J.~{Tabrikian}, ``Bobrovsky–{Z}akai-type
  bound for periodic stochastic filtering,'' \emph{IEEE Signal Process. Lett.},
  vol.~25, no.~10, pp. 1460--1464, 2018.

\bibitem{THANH_MISSP}
L.~T. Thanh, K.~Abed-Meraim, and N.~L. Trung, ``Misspecified
  {C}ram{\'e}r–{R}ao bounds for blind channel estimation under channel order
  misspecification,'' \emph{IEEE Trans. Signal Process.}, vol.~69, pp.
  5372--5385, 2021.

\bibitem{CCRB_misspecified}
S.~Fortunati, F.~Gini, and M.~S. Greco, ``The constrained misspecified
  {C}ram$\acute{\text{e}}$r-{R}ao bound,'' \emph{IEEE Signal Process. Lett.},
  vol.~23, no.~5, pp. 718--721, May 2016.

\bibitem{PREVOST_CHAUMETTE_RANDOM}
C.~Prevost, E.~Chaumette, K.~Usevich, D.~Brie, and P.~Comon, ``On
  {C}ram$\acute{\text{e}}$r-{R}ao lower bounds with random equality
  constraints,'' in \emph{Proc. of the IEEE International Conference on
  Acoustics, Speech, and Signal Processing (ICASSP)}, 2020, pp. 5355--5359.

\bibitem{JOURNAL_EYAL}
E.~Nitzan, T.~Routtenberg, and J.~Tabrikian, ``A new class of {B}ayesian cyclic
  bounds for periodic parameter estimation,'' \emph{IEEE Trans. Signal
  Process.}, vol.~64, no.~1, pp. 229--243, Jan. 2016.

\bibitem{ROUTTENBERG_BAY_PER}
T.~Routtenberg and J.~Tabrikian, ``Bayesian periodic {C}ram{\'e}r--{R}ao
  bound,'' \emph{IEEE Signal Process. Lett.}, vol.~29, pp. 1878--1882, 2022.

\bibitem{FAUSS_KULLBACK}
M.~Fau{\ss}, A.~Dytso, and H.~V. Poor, ``A {K}ullback-{L}eibler divergence
  variant of the {B}ayesian {C}ram{\'e}r–{R}ao bound,'' \emph{Signal
  Processing}, 2023.

\bibitem{PECARIC}
J.~E. Pecari{\'c}, S.~Puntanen, and G.~P.~H. Styan, ``Some further matrix
  extensions of the {C}auchy-{S}chwarz and {K}antorovich inequalities, with
  some statistical applications,'' \emph{Linear algebra and its applications},
  vol. 237, pp. 455--476, 1996.

\bibitem{CONVEX_BOYD}
S.~Boyd and L.~Vandenberghe, \emph{Convex optimization}.\hskip 1em plus 0.5em
  minus 0.4em\relax Cambridge university press, 2004.

\bibitem{HENDERSON_VEC}
H.~V. Henderson and S.~R. Searle, ``The vec-permutation matrix, the vec
  operator and {K}ronecker products: A review,'' \emph{Linear and multilinear
  algebra}, vol.~9, no.~4, pp. 271--288, 1981.

\end{thebibliography}
